\documentclass[12pt,reqno]{amsart}
\usepackage{amsthm,amsmath,amssymb,mathrsfs,graphicx,caption, subcaption, color,csquotes,verbatim,url,hyperref,bookmark,algorithmic}

\usepackage[a4paper,margin=1.25in,footskip=0.25in]{geometry}

 \graphicspath{{./fig2/}}

\usepackage{lineno}
\usepackage{setspace}
\numberwithin{equation}{section}

\usepackage{cite}

  \usepackage{longtable}
  \usepackage{color}
  \usepackage{colortbl}

 \definecolor{textcol}{rgb}{.118, .565, 1.00}
 \definecolor{rowcol}{rgb}{.218, .565, 1.00}

 {

\newtheorem{assumption}{Assumption}

\newtheorem{definition}{Definition}

\newtheorem{remark}{Remark}

\newtheorem{proposition}{Proposition}
\newcommand{\nn}{\mathbb{N}}
\newcommand{\zz}{\mathbb{Z}}

\newcommand{\pp}{\mathbb{P}}

\newcommand{\indi}{\text{\bf 1}}

\newcommand{\set}[1]{\left\{#1\right\}}

\newcommand{\Rset}{\mathsf{R}}
\newcommand{\Btn}{\mathbb{B}}

 \renewcommand{\d}{\ensuremath{\,\mathrm{d}}}

\renewcommand{\vec}[1]{\boldsymbol{#1}}
\newcommand{\po}{\vec{\xi}}
\newcommand{\pop}{{\xi}}

\newcommand{\pw}{\vec{\omega}}

\newcommand{\pt}{\vec{\varphi}}

\newcommand{\dom}{{D}}

\newcommand{\Lr}{{L}}

\newcommand{\leqref}[1]{Eqn~\eqref{#1}}
\DeclareMathOperator{\argmax}{\arg\max}
\DeclareMathOperator{\Cp}{\mathsf{Cp}}

\newcommand{\FW}{{Freidlin-Wentzell}}
\newcommand{\tpt}{{transition path theory}}
\newcommand{\Lm}{{logistic map}}

\makeatletter
\g@addto@macro{\endabstract}{\@setabstract}
\newcommand{\authorfootnotes}{\renewcommand\thefootnote{\@fnsymbol\c@footnote}}%
\makeatother

\begin{document}
\title{ }

\begin{center}
  \LARGE
  Explore Stochastic Instabilities
  of Periodic Points
   by Transition Path Theory
   \par \bigskip

  \normalsize
  \authorfootnotes

     Yu Cao\textsuperscript{1},
  Ling Lin\footnote{L. Lin acknowledges the financial
  support of the DRS Fellowship Program of Freie Universit\"{a}t Berlin.}\textsuperscript{2},
  %\footnote{Current   address: Institute for Mathematics, Freie Universit\"{a}t Berlin,
   %Arnimallee 6, 14195 Berlin, Germany.},
  Xiang Zhou\footnote{Corresponding author. email: {\it xiang.zhou@cityu.edu.hk}.
X. Zhou acknowledges the financial
  support of Hong Kong GRF (109113,11304314,11304715).
  %Corresponding address:
 % Department of Mathematics, City University of Hong Kong, Tat Chee Ave, Kowloon, Hong Kong SAR.
   }\textsuperscript{1},

   \par

    \bigskip

\textsuperscript{1}
Department of Mathematics,
        City  University of Hong Kong,  \par Tat Chee Ave, Kowloon,  Hong Kong.\par
  \textsuperscript{2}  Institute for Mathematics, Freie Universit\"{a}t Berlin,
  \par  Arnimallee 6, 14195 Berlin, Germany.
  \par
    \bigskip

  \today
\end{center}

\date{\today}
\maketitle
	
\begin{abstract}
We consider the  noise-induced transitions
in the randomly perturbed discrete  logistic map
 from a linearly stable periodic orbit consisting of $T$  periodic points.
 The traditional large  deviation theory and
asymptotic analysis for small noise limit
as well as  the derived quasi-potential can not
distinguish the quantitative difference in  noise-induced stochastic  instabilities   of these
 $T$  periodic points.
 We  generalize  the transition path theory to the discrete-time continuous-space
 stochastic process to attack this problem.
As a first criterion  of quantifying  the relative instability
 among $T$ periodic points,
we compare the distribution of the last passage locations
in the transitions from the whole periodic orbit to a prescribed set far away.
This distribution is related to  the contributions to the transition rate
from each periodic points.
The second criterion is based on the competency of the transition paths
associated with each periodic point.
Both criteria  utilise   the reactive  probability current
in the transition path theory.
 Our numerical  results for the logistic map  reveal
 the transition mechanism of escaping from the stable periodic orbit
 and  identify
  which     periodic point     is more prone to lose stability so as   to make successful  transitions
  under random perturbations.	
\end{abstract}

\maketitle

%\tableofcontents

\keywordsname{: random logistic map,  transition path theory, periodic orbit, stochastic instability}

\section{Introduction}

When deterministic dynamical systems are perturbed by random noise,
even though the noise amplitude is small, it has a prominent influence on the dynamics on the appropriate time-scale.
For example, the thermal noise can induce    important
physical and biological metastability  phenomena  such as   phase transitions,  nucleation events,
configuration changes of macromolecules.
These phenomena correspond to the very unlikely
excursions  in the phase space of the  random  trajectories,
so these events  are  usually called {\it rare events}.
These   trajectories have to
 overcome some barriers
 to escape from the initial metastable state and enter another.
For the  understanding of  the occurrence of rare events,
it is of great importance    to
investigate  the non-equilibrium  statistical and dynamical  behaviours
of those  trajectories successfully making   transitions.
One of the interesting  questions is
how the ensemble of these transition trajectories
depend on  the phase space  of the unperturbed deterministic dynamical systems,
for example,  what structure  in the phase space would be the  barriers for  transitions,
how the system leaves the initial metastable state and
escapes  the basin of attraction
of this  metastable state, etc.
For  general dynamics, the metastable state
may not be a single point as a local minimum on potential energy surface;
it may be a collection of points, such as
 limit cycle, periodic orbit, or even  chaotic attractor.
 In this paper,  we are interested in,
 conditioned on the   occurrence of   rare transitions
 from  one of these  stable structures,
   through   which location {\it within}  the   metastable set
     the transition trajectories will  leave  with a higher or dominant probability.
Particularly, as an example, our study focuses on   the stable periodic orbits
in the randomly perturbed    logistic map.

In history,  many  research work target    to explore  the barrier on the basin boundary.
For the diffusion process on  a
potential energy surface (a classic model
for chemical reactions\cite{Kramers,Kampen}),
the well-known transition-state theory\cite{Erying} states that
basically   the transition state,   is a saddle point with index 1 on the potential energy surface.
The   progresses of chemical reactions are mainly described by
 heteroclinic orbits connecting the local minima
through  the saddle point, i.e., ``minimum energy path''.
In addition, one can calculate  the transition rate by computing the probability flux of particles that cross the dividing surface of two neighbouring
 potential wells.
For general continuous time dynamical systems under random perturbations,
the notion of ``most probable     path" is very useful to describe the transition process.
This path    is a
curve in the phase space with a dominant contribution in the ensemble
of transition trajectories at vanishing noise limit.
From a mathematical  viewpoint,
such  a notion  of most probable path is based on the large deviation principle (LDP)
in path space. The well-known \FW\ theory \cite{FW1998}
states that the most probable transition path from  one set  $A$ to another  $B$
 is the {\it  minimum action path}, which minimizes
the rate function  of the \FW\  LDP (aka.``\FW\ action functional")
subject to the constraint  of starting from $A$ and ending at $B$.
The  transition probability  is dominated by the minimal
value of the rate function.
Therefore, by
analytically  performing asymptotic analysis
such as WKB or   instanton  analysis
  \cite{MatSchuss1983, Naeh1990,Maier1992PRL,Maier1993PRE},
  or numerically solving the
variational problem in a path space
\cite{weinan-MAM2004,aMAM2008,Heymann2006},
one can identify
  most probable escape/transition path. This allows a
further examination of the path and the
unstable structure  in the phase space, in particular,
how this path crosses the basin boundaries.
This methodology of least action principle
is  applicable for   general   dynamical systems of
 continuous-time or discrete-time.
The applications to Lorenz model\cite{Lorenz2008}, Kuramoto-Sivashinsky PDE\cite{KS-WZE2009}
 have already discovered  the barriers on the basin boundary
in types of saddle points or  saddle cycles.

For discrete maps perturbed by noise, there
has been a long history of studying the effect of  random perturbations
on the dynamics. Some works are based on the brute-force simulation
to collect the empirical  distributions of transition
trajectories\cite{DykmanPRL1992}.
The applications of  the large deviation rate function
in the setting of discrete-time maps included
 the work in \cite{KautzPRA1987,KautzPRA1988}  which studied   the transitions between stable fixed points,  stable periodic orbits and chaotic attractors, providing empirical  evidence that
the transition state is the type of  a saddle
node.
\cite{Graham1991,Kraut2003}
focused on  the {\it quasi-potential}
(activation energy),
which is a good  quantification of the stochastic stability
for a metastable set,
 to investigate the key invariant set on the basin boundary.
The series work of \cite{LuchJETP1999,LuchPRL2003,LuchPRE2005}
carried extensive studies for   Lorenz systems, Henon maps and  other examples of discrete maps
under additive random perturbation.
Their results seem to suggest that  in
the noise-induced escape from the basin of attraction of a
stable set,  the barrier-crossing on the basin boundary
is mostly  determined by the position and stability
properties of  certain saddle point or saddle  cycles.
Recently, a new approach was developed
in \cite{Billings::virus2002,Bollt2002153}  to understand transport in stochastic dynamical systems.
They basically use the transition probability matrix
(after discretizing and reindexing the continuous   space) for identification   of active regions of stochastic transport.
Most of these existing studies deal with
the transition state (or the set)
on the basin of attraction of a metastable state (or invariant set).

In this paper,  we are interested in the transition from set to set
with the purpose of pinpointing  the role of individual points in the
  initial metastable set to escape.
   The motivation comes from the questions  below:
     how the randomly perturbed system
   leaves the periodic orbit (or limit cycle in continuous time dynamics);
   how  the stable self-sustained oscillating motion is eventually destroyed by the noise.

   Specifically,  we consider the random logistic map with additive Gaussian noise.
   We are concerned with the noise-induced transitions from $A$ to $B$ ---  two disjoint sets in the phase space.
   It is assumed that the unperturbed system has
 a linearly  stable periodic orbit (all eigenvalues are less than one in modulus), denoted as  $\po=(\pop_1,\pop_2,\cdots,\pop_T)$, where the integer $T$
is the period.
To explore the stochastic instability of
$\po$, we select $A$  as the union of the $T$ periodic points $\set{\pop_i}$
(more precisely, $A$ is the union of $T$ small windows around  $\set{\pop_i}$.
Refer to Section \ref{sec:log}).
After an exponentially long time wandering
around  the metastable set  $A$ in  the random motion of
nearly periodic oscillation,
the stochastic system will eventually  get a chance of
making a significant  transition to a  set  $B$  far away from $A$.
The question we shall address is
how the system deviates from the typical periodic oscillation and
whether it have any preference to some special periodic points to
make the transition.

The traditional techniques   based on
  large deviation principle  and the concept of quasi-potential
are not capable of addressing  the above question  due to the following fact, although they
are quite successful in studying the most active regions  {  on the basin boundary} of  the set $A$.
 If the  unperturbed deterministic
flow can go from a point  $x$ to another point $y$,  then
the cost (quasi-potential) from $x$ to $y$ is simply zero.
Thus,  if   any points in the  set $A$ can reach each other mutually  by the  deterministic
flow (periodic orbit or limit cycle certainly satisfies  this condition), then
the quasi-potential is flat on the whole set $A$.
 From any point in $A$, the minimal action to escape the basin
is the same.  The extremal path   minimizing the action functional
usually  takes infinite time
 and has    infinitely length, and  the whole invariant set  $A$ is the  $\alpha$-limit set of the extremal path:
There is no particular location in the set $A$ from where
the extremal path emits.
Hence, the action functional and the minimum action path can not distinguish individual points
inside  $A$ in such cases.
Similarly,
 the singular perturbation method
 \cite{MatSchuss1982} for the mean first passage time in the vanishing noise limit
will give a constant value of the WKB solution on the
stable limit cycle, and
thus may be not directly useful to our problem.

We use a new and attractive tool,  the   transition path theory \cite{Weinan::towards2006, Eijnden::tpt2006,Metzner::mjp2009,
TPT2010},
by modifying  this theory  for the discrete map.
The \tpt\ for continuous-time dynamical systems
 has been   proved to be an effective mathematical tool to reveal
 transition mechanism of a few complex physical and biological systems \cite{Noe::folding2009,Cameron::LJcluster2014}.
This article intends to bring the transition path theory
into studying the stochastic instability issues  for  random  discrete maps.
We shall formulate    the \tpt\ for the discrete-time continuous-space Markov process.
  We then use three key ingredients in the \tpt,
the reactive current, the transition rate and  the dominant transition path,
    to understand
the escape mechanism from the periodic orbit $A$ for any finite noise.
To quantitatively compare the stochastic instability
of the $T$ individual periodic points, we propose
two rules: the first one is the distribution of the last passage position
among these $T$ point and the second one is the starting point
of the dominant transition path.
Our  numerical  results obtained clearly show the capability of this theory in
quantitative  understanding of
the different roles of the individual points belonging to the same periodic orbit.

The paper is organized as follows.
 In Section \ref{sec:log}, we will set up  our  problem for the random logistic map. In Section \ref{sec:method},
  we briefly  review the existing methodologies.
  Section \ref{sec:TPT} is our method based on the \tpt.
  In section \ref{sec:num}, we   present   numerical results for the random logistic map.
  Section \ref{sec:sum} is our concluding discussion.

\section{Random Logistic Map}
\label{sec:log}

The randomly perturbed discrete map of our interest is the following
\[ x_{n+1} = F(x_n) + \sigma \eta_n\]
where
   $\eta_n\sim N(0,1)$  are {\it i.i.d.} standard normal random variables
   and the constant $\sigma>0$ is the noise amplitude.
 In this paper, we focus on a well-known example of $F$: the \Lm.
Logistic map is probably the simplest nonlinear mapping
giving rise to periodic  and chaotic behaviors.
It is popularly used as a discrete-time demographic model
to represent the population with density-dependent mortality.
 Mathematically, the logistic map is written by
 \[ x\to  F(x):=\alpha x(1-x),\]
where  $x$ is a number between zero and one that represents the ratio of existing population to the maximum possible population.
 $\alpha>0$ is the parameter.
When  $x$ is out of the interval $[0,1]$, the logistic map simply diverges to infinity and never returns.     The dynamics   of interest is in the interval $[0,1]$.
There are two fixed points in this interval, $0$ and $1-\frac{1}{\alpha}$.
When $0<\alpha<1$, $0$ is the only stable fixed point and when $1< \alpha < 3$, $1-\frac{1}{\alpha}$ is the only stable fixed point. Both fixed points become unstable for  $\alpha$ larger than $3$.
$\alpha = 3$ is the onset of a stable period-2 orbit, and
this period-2 orbit disappears    at $\alpha = 1+\sqrt{6}\approx 3.4495$, at which the period-4 orbit takes over.
The stable period-$2^{n}$ orbit is followed by the stable period-$2^{n+1}$ orbit if $\alpha$ increases continuously. This phenomenon  is termed as period doubling cascade
and leads to the onset of chaos. Apart from this, tangent bifurcation is  found,
e.g.,   the onset of stable period-3 orbit arises at  $\alpha = 1+2\sqrt{2} \approx 3.828$.
 Further details about the logistic map can be found in some classic literature, e.g.,  \cite{Ott::chaos1993}.

The random logistic map of our interest is
   the following additive random perturbation
restricted on  the interval $\dom=[0,1]$ with $F(x)=\alpha x(1-x)$:
\begin{equation}
\label{eq:Lm}
 x_{n+1} = F(x_{n}) + \sigma \eta_n  \mod 1.
 \end{equation}
We here impose the periodic boundary condition
for the Markov process $\{x_n\}$  so that
all  dynamics is restricted on the compact set  $\dom$.
This will guarantee
the
unique existence of the invariant
measure for $\{x_n\}$ on $\dom$
and thus  ergodicity holds for this stochastic process,
which is a fundamental assumption in the \tpt.
Other type of boundary condition is also feasible
such as the reflection boundary condition at $x=0$
and $x=1$.

   The  transition probability density of
   the discrete-time continuous-space Markov process   \eqref{eq:Lm}
      is \begin{equation}
P(x,y) = \sum_{l\in \zz} \frac{1}{\sqrt{2\pi\sigma^2}}
\exp\left (-\frac{1}{2\sigma^2}(y-F(x)+l)^2\right),
\end{equation}
where the sum over the integer  $l$ is merely a minor adjustment for the periodic boundary condition we used here.
The   density of the unique invariant measure, $\pi(x)$,  is the solution
of the following balance equation
\begin{equation}
\label{eqn:inv}
\int_{\dom} P(y,x) \pi(y) \d y =\pi(x).
\end{equation}
In other words, $\pi(x)$ is the  eigenfunction
for the principle eigenvalue of the adjoint of the transition kernel $P(x,y)$.

Now we specify the sets involved in
 the transition problems for the randomly perturbed \Lm.
 The parameter $\alpha$ in our study will be selected so that
the \Lm\ only has periodic oscillations.
The stable invariant set   of interest here
is the (linearly)   stable period-$T$ orbit
in the (unperturbed)   \Lm,
denoted as $\po = (\pop_1,\pop_2,\cdots, \pop_T)$.
 The order of $(\pop_i)$ in $\po$ is
specified  so that $\pop_{i+1}=F(\pop_i)$.
We   pick  a small neighbourhood  $A$ around the $T$
periodic points and a   disjoint  set $B$.
With these setups,  the noise-induced transitions from $A$ to $B$, named as $A$-$B$ transitions,
  will be our focus.
By specifying the width $\delta_a$, the set $A$ around the periodic orbit $\po$ is
the union of the $T$ disjoint small windows
\begin{equation}\label{setA}
A = \underset{1\le i\le T}{{\cup}} [\pop_i-\delta_a,\pop_i+\delta_a].
\end{equation}
It is possible to specify different widths  for different periodic points,
or let the interval be asymmetric around $\pop_i$.
It is also  possible to use the level set of the invariant measure $\pi$,
$\set{x: \pi(x) < \delta}$,
around the periodic points.
In the study of this paper, we use the same $\delta_a$ for simplicity.
The set $B$ is placed near the unstable fixed point $0$ (or $1$)
with the width $\delta_b$:
 \begin{equation}\label{setB}
 B = [0,\delta_b]\cup [1-\delta_b,1].
 \end{equation}
 $\delta_a$ and $\delta_b$ are small enough so that $A\cap B = \emptyset$ and $[\pop_i-\delta_a,\pop_i+\delta_a]\cap [\pop_j-\delta_a, \pop_j+\delta_a] = \emptyset$ is empty  for any $1\le i < j\le T$.
 The set $B$ in our logistic map example is around the unstable point,
 the ``furthest" boundary point from the stable set $A$.
 In general situations, this set $B$ is placed just outside  the basin of attraction of the periodic orbit $\po$ and the
  instability result about the $\po$ is typically robust for small noise amplitude.

We introduce the nonzero width $\delta_a$ for the periodic orbit $\po$,
because the space is continuous, not discrete:
it makes no sense to consider    trajectories in stochastic system exactly
leaving or entering  some singleton  points, at a fixed noise amplitude $\sigma>0$.
In practice, the  specification  of the window width $\delta_a$
should be  given by the user    who decides  to what extent
the system is deemed as out of the   oscillation status for specific applications.

Usually, the   width $\delta_a$  should be small enough so that
 the set $A_i$ can represent the transition behaviour  for the point $\pop_i$ inside.
  In theory,  for  a set $A$ to  truly
reflect the transition  mechanism
of escaping from $\po$,
the width $\delta_a$ should approach zero.
In fact,  all calculations are based on a finite $\delta_a$.
 But  since  the   set
$A$  has the  metastability property (linearly stable), then  it follows that
the conclusions to our question  based on the study of the set $A$
for finitely small $\delta_a$
are quite  robust and indeed give correct insights about the
  transition mechanisms  and the stochastic instabilities
  for the stable periodic orbit $\po$.

\section{ Related  Works    }
\label{sec:method}

We first briefly  review two existing  methods
for  the study of   stochastic systems.
The known applications of both methods
are mainly  for exploring the basin boundary.

\subsection{Large deviation principle}
 We give a glimpse of the large deviation principle  (LDP)
 or  the principle of least action for randomly perturbed discrete map.
 For continuous-time diffusions processes, refer to  the \FW\ theory
in  \cite{FW1998}.
We start from the transition probability for the random mapping $x_{n+1} = F(x_n) + \sigma \eta_n$,
which is
$$P(x,y) = \frac{1}{\sqrt{2\pi \sigma^2}} \exp \left( -\frac{(y-F(x))^2}{2\sigma^2} \right).$$
With the fixed initial point $x_0$ at time $0$ and ending point $x_n$ at time $n$, the probability of a  path $\vec{\gamma} = (x_0, x_1, \cdots, x_{n-1}, x_n)$ is
\begin{equation}
P[\vec{\gamma}] = \prod_{i=0}^{n-1} p(x_i, x_{i+1}) \propto  Z_\sigma^{-1}
\exp\left(-\frac{1}{\sigma^2} S[x_0, \cdots, x_{n} ]\right),
\end{equation}
where  $Z_\sigma^{-1}$ is the prefactor and the cost function $S$ has  the form of
\begin{equation}
S[\vec{\gamma}]=S[x_0, \cdots, x_{n}] = \frac{1}{2}\sum_{i=0}^{n-1} (x_{i+1}-F(x_i))^2.
\end{equation}
This cost function $S$ is actually   the {\it rate function} (aka. action) of the LDP
at the vanishing noise limit $\sigma\downarrow  0$.
By the Laplace's method,
the path probability  $P[\vec{\gamma}]$ is asymptotically dominated by
$\exp\left( -\frac{1}{\sigma^2}  S_{\min}\right),$
where $S_{\min} = \min_{\gamma} S[\vec{\gamma}]$.
The {\it minimum action  path}  (MAP) $\vec{\gamma}^{*}$
 is the one such that    $S[\vec{\gamma}^{*}] = S_{\min}$.
If this minimal action $S_{\min}$ is viewed as the function of
the initial point $x_0$ and the ending point $x_n$ for all possible $n$,
then it is the so-called {\it quasi-potential},
which is quite useful for quantifying
the stability of each basin   against the random perturbation
\cite{FW1998,KautzPRA1987,subcrit2010}.
When the initial point $x_0$ is in a stable structure
(fixed point, periodic orbit, chaotic attractor) of the phase space,
and $x_n$ is out of the basin of attraction of this stable structure,
 the MAP is usually called the  {\it most probable escape path} (MPEP).
 The intersection part of the MPEP with the basin boundary
 is quite revealing for transition states or active regions during  crossing the boundary.

%
%There is an elegantly Hamiltonian dynamics
%corresponding to the action $S$ above.
%Thus, the minimization problem $\min_\gamma S[\gamma]$
%can be also solved via the boundary value problem
%in the extended Hamilton system.
%In the discrete-time setting,
%it is also possible to track the unstable manifold
%of the Hamilton system for simple examples.
%Refer to \cite{Luchinsky::bvp2005} for such a study
%of the MPEP for the Henon map
%starting from one stable
%period-1 orbit.
%Many interesting results of the MPEP
%on the basin boundary have been investigated
%to reveal the importance role of many saddle structure,
%such as \cite{Lorenz2008}.

One obvious feature  of this least action method
based on the LDP
is that the cost is {\it zero} for a path
from $\pop_1$ to $\pop_2$ if $\pop_2$ is exactly equal to
 $F(\pop_1)$. This means that starting from
 any point in the same period-$T$ orbit,
 the minimal action is the same.
Thus, one can not tell which point in the periodic orbit,  limit cycle,
or even  chaotic attractor, is more prone to
  the random perturbation, since they share the   same
action.

\subsection{PDF flux }
\label{ssec:pdfflux}
To study  the bi-stabilities in  the stochastically perturbed dynamical systems,
Billings {\it et al.} \cite{Bollt2002153,Billings::virus2002}
  proposed
a method  on the   transport  of probability density function
under the discrete map,
in which the one-step transport
is described by the Frobenius-Perron operator,
i.e., the adjoint of the transition kernel    $P(x,y)$.
They investigated how an initial distribution
is transported to a given region in the phase space under the
action of this operator:
\[ \rho(x) \to \mathcal{F}[\rho](x):= \int_\dom P(y,x) \rho(y)\d y. \]
 Depending on the initial distribution,  they call
$\mathcal{F}[\rho]$ the {\it area flux}
if $\rho$ is uniform and  call $\mathcal{F}[\rho]$ the  {\it PDF flux} if $\rho$ is  the invariant measure
$\pi$ (\leqref{eqn:inv}). For a given set $A\subset \dom$,
 the ``mass flux into $A$"
 is defined as
 \[ \int_{x\in A} \left( \int_{ y \in \dom\setminus A} P(y,x) \rho(y)\d y \right) \d x \]
and ``mass flux out of  $A$ "  (by switching $A$ and its complement set $\dom\setminus A$)
 is defined as
  \begin{equation}
  \label{eqn:Fa}
  \begin{split}
\mathcal{F}^-_A
 & = \int_{x  \in \dom\setminus A} \left( \int_{ y \in  A} P(y,x) \rho(y)\d y \right) \d x
   \\
 =&   \int_{y  \in \dom\setminus A} \left( \int_{ x \in  A} P(x,y) \rho(x)\d x \right) \d y
 \\
 =&   \int_{ x \in  A} \rho(x) \left( \int_{y  \in \dom\setminus A}   P(x,y) \d y \right)   \d x.
    \end{split}
    \end{equation}
 The quantity $\rho(x)P(x,y)$ was used
for $x$ in one basin and $y$ in another basin
to investigate  where a trajectory is most likely to escape    the basin boundary.
For a few applications \cite{Billings::virus2002}, the saddle cycles on the basin boundary
usually have the maximal flux across the boundary.

 \section{ Transition Path Theory for discrete map }
 \label{sec:TPT}

We first formulate the \tpt\ for discrete map.
Then we identify the point in the orbit $\po$ with the highest probability
mass of being the last passage position during the $A$-$B$ transition,
which is actually the point with the biggest contributions to the  transition rate.
To further study the  development of the current for the transition probability
after emitting the set $A$,  we  carry out
the pathway analysis and target for the dominant transition paths.
The precise definitions of these concepts will be explained soon.
 We remark that the first approach based on the transition rate
 is  relatively easy for calculation and quite universal for any situations.
 The second path-based approach needs a thorough exploration of
connected paths based on network   theory and  could  have  difficult  situations
 that fail to compare the stochastic instability in a quantitative way
 due to complexity of pathways,
 although our \Lm\ example does not meet with  such dilemma
 and shows a clean result.
 In addition, the above  two approaches may also give two different conclusions
 since the viewpoints of interpreting and comparing   the stochastic instabilities are different.

\subsection{Transition path theory for randomly perturbed discrete map}

The original TPT was formulated for the continuous-time
continuous-space Markov process \cite{Weinan::towards2006,Eijnden::tpt2006,TPT2010}.
The TPT for the continuous-time discrete-space Markov process
(jump process) was developed in \cite{Metzner::mjp2009},
in which a detailed analysis for the pathways on the discrete space
is of particular interest.
Here we present the method of  TPT in the setting of
the discrete-time
continuous-space Markov process.

 The \tpt\  does not consider the limit of vanishing noise.  It assumes that the stochastic system
is ergodic and has a unique invariant measure.
The main focus  of the
TPT is  the statistical behaviour of the
ensemble of    {\it reactive trajectories}
between two arbitrary disjoint sets.
Assume that  $A$ and $B$
are two disjoint  closed subsets of
the state space $\dom$ ($\dom=[0,1]$ for our example of the logistic map),
each of which is the closure of a nonempty
open set.
The transition of our interest is from $A$ to $B$.
For a discrete-time homogeneous Markov process
$\{X_n: n \in \zz  \}$,   define
the first   hitting time after time $m$ and the last hitting time before time $m$ of $A\cup B$ as follows, respectively,
\begin{equation}
\begin{split}
H^{+}_{AB}(m) &:= \inf\{n\ge m: X_n \in A\cup B \},\\
H^{-}_{AB}(m) &:= \sup\{n\le m: X_n \in A\cup B \}.\\
\end{split}
\end{equation}
Then for a generic trajectory $(X_n)_{n\in \zz}$, the ensemble of
{\bf $A$-$B$ reactive trajectories} is defined to be
the collection of     pieces of  the truncated trajectories:  $\{X_n: n\in \Rset\}$,
where $n\in \Rset$ if and only if
\[X_{H^+_{AB}(n+1)} \in B~~~\mbox{   and   } ~~~~X_{H^-_{AB}(n)} \in A. \]
$\Rset$ is the set of times at which $X_n$ belongs to an $A$-$B$ reactive trajectory.
Refer to Figure \ref{fig:traj} for one piece of reactive trajectory
extracted from a generic trajectory.
The intuition for defining   $A$-$B$ reactive trajectories  is
that the points  on these reactive trajectories will  first reach $B$  rather than $A$
and came from  $A$ rather than $B$.

\begin{center}
\begin{figure}[htbp]
\includegraphics[width=0.91\textwidth]{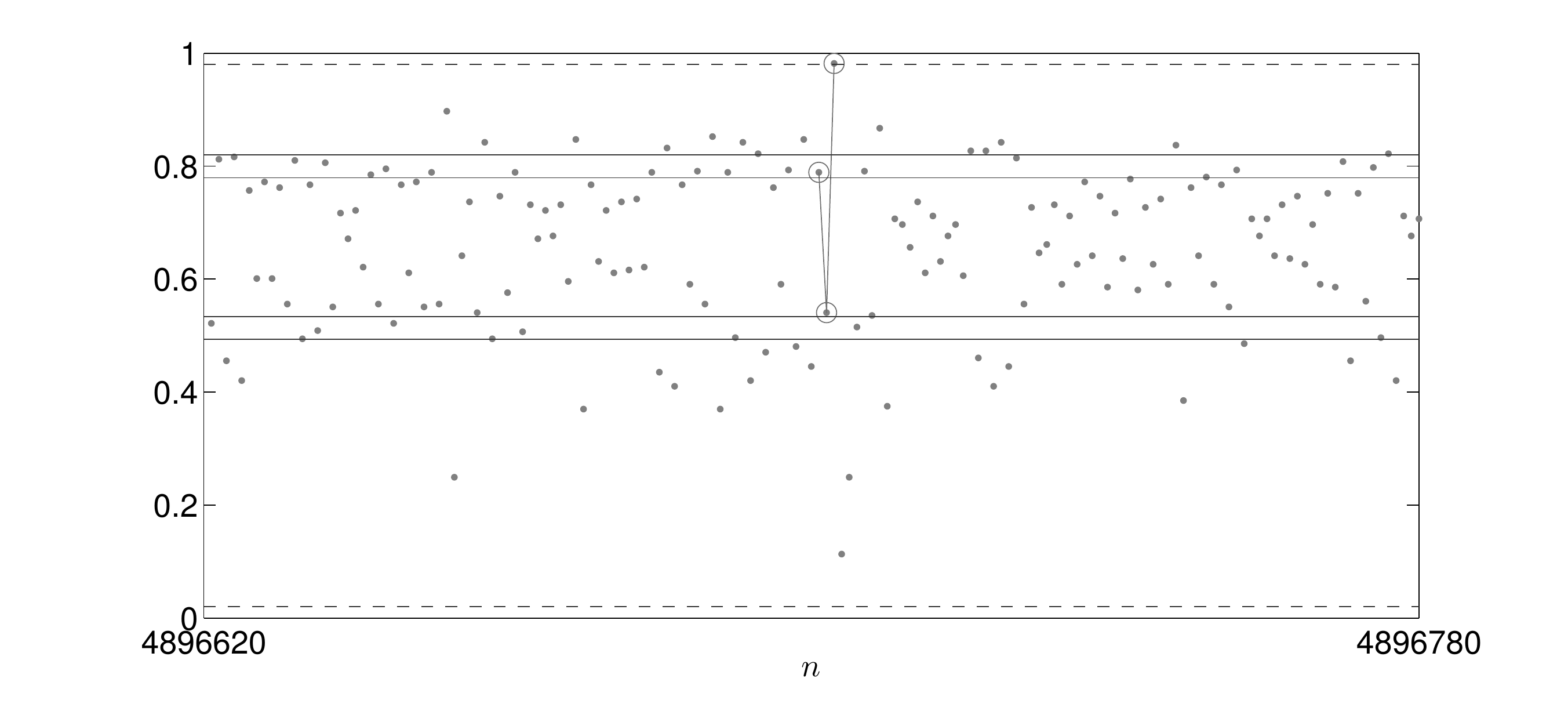}
\caption{A snapshot of  a generic trajectory (dots in the plot) and one reactive trajectory (three points circled in the plot)  of the randomly perturbed \Lm  observed
in the time interval $[4896620, 4896780] $. The set $A$ is the union of $A_1$ and $A_2$
around the periodic points $\po=(0.5130, 0.7995)$, corresponding to
two narrow bands with length $2\delta_a= 0.04$  shown   by solid horizontal lines.
The bounds of the set $B$ near $0$ and $1$ are shown in dashed lines.
$\alpha=3.2$, $\sigma=0.04$.
}
\label{fig:traj}
\end{figure}
\end{center}

The most important ingredient in the TPT is the  probability
 current for  $A$-$B$ reactive trajectories.
For the  continuous state space $\dom$, we introduce
its space-discretized version first:
\begin{equation}
\begin{split}
\mathsf{J}(x,y,\Delta x,\Delta y):=\lim_{N\rightarrow \infty}\frac{1}{2N+1}\sum_{n=N}^{-N}
\bigg(  \indi_{[x-\frac{\Delta x}{2}, x+\frac{\Delta x}{2}]}(X_n) \indi_{[y-\frac{\Delta y}{2}, y+\frac{\Delta y}{2}]}(X_{n+1})
\\
\indi_{A}(X_{H^{-}_{AB}(n)}) \indi_{B}(X_{H^{+}_{AB}(n+1)})
 \bigg),
\label{empirical_F_x_y}
\end{split}
\end{equation}
where $\indi_{\{\cdot\}}(\cdot)$ is the indicator function. Then the
{\bf $A$-$B$ reactive probability current}
  is  defined as the following limiting function for $x$ and $y$ in $\dom$,
\begin{equation*}
J(x,y) := \lim_{\Delta x,\Delta y\to 0} \frac{\mathsf{J}(x,y,\Delta x, \Delta y)}{\Delta x \Delta y}.
\label{def_J_x_y}
\end{equation*}
We sometimes just call $J$   the reactive current  whenever the
specification of the sets $A$ and  $B$ is clear.

The above definition of the reactive current  $J$ is
based on the time average for an infinitely long generic  trajectory.  To obtain an
ensemble average, we need assume the
Markov process $\{X_n\}$ is ergodic, i.e.,
the unique existence of the invariant probability density such that
$\pi(x) = \underset{N\rightarrow \infty}\lim \frac{1}{N}\sum_{n=0}^{N-1} \indi_{x}(X_n)$.
Then, \eqref{empirical_F_x_y} leads to the following formula  of the reactive current
\begin{equation}\label{J_equation}
J(x,y) = q^{-}(x)\pi(x)P(x,y)q^{+}(y),~~x\in \dom, y \in \dom.
\end{equation}
where $P(x,y)$ is the transition density function of  the Markov process
$P(x,y\d y = \pp[ X_{n+1} \in [y,y+\d y)\vert X_n=x ] $,
$q^+$ and $q^-$ are the
  the forward and backward  committor functions,  defined  as follows, respectively:
\[
\begin{split}
q^{+}(x) := \mathbb{P}[X_{H^{+}_{AB}(0)} \in B\vert X_0=x],
~~~~
q^{-}(x) := \mathbb{P}[X_{H^{-}_{AB}(0)} \in A\vert X_0=x].
\end{split}\]
%where $\mathbb{P}_x$ means the probability law  conditional on   $X_0=x$.
By definition, the committor functions satisfy the following boundary conditions
\begin{equation}
\label{eqn:bc}
\begin{cases}
  q^{+}(x) = 0, ~\mbox{and }~q^-(x) =1,   & \mbox{if}~~ x\in A,  \\
  q^{+}(x) = 1, ~\mbox{and }~q^-(x) =0,   & \mbox{if}~~ x\in B.  \\
    \end{cases}
 \end{equation}
 This implies the fact
\begin{equation}\label{eqn:J=0}
J(x,y)=0,~~ \mbox{when }  x\in  B, y\in \dom \mbox{  or  }  x\in \dom,  y \in A.
\end{equation}
It is known from  \cite{Weinan::towards2006,Metzner::mjp2009} that
the committor  functions satisfy the following Fredholm integral equation
for all $x \notin A\cup B$,
\begin{equation}
q^{+}(x) = \int_{\dom}  P(x,y) q^{+}(y)\d y,  ~~~ x \in \dom\setminus (A\cup B)
\label{eqn:q+}
\end{equation}
and
\begin{equation}
q^{-}(x) = \int_{\dom} P^{-}(x,y) q^{-}(y) \d y,~~~ x \in \dom\setminus (A\cup B)
\label{eqn:q-}
\end{equation}
where \begin{equation}
P^{-}(x,y): = \frac{1}{\pi(x)} P(y,x)\pi(y)
\label{bk-P}
\end{equation}
is the transition kernel of the time reversed
process $\{X_{-n}\}_{n\in \zz}$.
Since the transition kernel $P$ is irreducible in the ergodicity assumption,
then the functions $q^+(x)$ and $q^-(x)$ are always strictly positive
for any $x\notin A\cup B$.

\begin{remark}\label{rk:pdfflux}
Compared with the   PDF flux $\pi(x)P(x,y)$
in Section \ref{ssec:pdfflux},
the $A$-$B$  reactive current $J(x,y)=\pi(x)P(x,y)q^-(x)q^+(y)$ in  the \tpt\
includes the additional  global information for the $A$-$B$ transition
of   the committor functions.  These two quantities  are equal
only when $x\in A$ and $y\in B$.
%For any $x\notin A\cup B$, the committor functions give the conditional probability that the process came from $A$ last and enter $B$ first.
%
%
%Note that when $x\in A$,  $q^-(x)=1$.
%Usually $A$ is chosen as a stable set  for the deterministic system
%(a metastable ``well").
%Then, at the boundary of the basin of attraction of $A$,
%   $q^+(y)$ is quite close to $0.5$
%for a set  $B$  out of  the basin of attraction of $A$
%and in a neighbouring basin.
%Consequently,  $\pi(x)P(x,y)$ and $J(x,y)$ are
%not significantly  different for $y$ on the basin boundary.
%Hence, it is expected
%that both formulations would identify
%the same barriers such as some saddle points or saddle cycles
% on the basin  boundary,
%in particular when  the   noise  amplitude is small.
%But for our problem of  distinguishing the
%individual periodic points
%in a stable periodic orbit $\po$, which is selected as the set $A$,
%these two formulations are really different.
\end{remark}

In the next,  we shall address two main issues about the methods
based  on the TPT
for the application to the random perturbed discrete map.
The first one  is the robust  calculation of the reactive current  function
$J(x,y)$ and the second   is how to use this reactive current function to
analyze the reaction pathways as well as the reaction rate.
Based  on these developments, we shall carry out the study  for the roles  of individual points
in $A$ and   evaluate their stabilities in
the content of $A$-$B$ transitions.

We   rewrite the equations \eqref{eqn:q-} and \eqref{bk-P} by introducing $\tilde{q}^-(x) := \pi(x)q^{-}(x)$, then
\begin{equation}
\label{eqn:q-t}
\tilde{q}^-(x) = \int_{\dom}  P(y,x) \tilde{q}^-(y)\d y,
~~x\in \dom\setminus (A\cup B).
\end{equation}
The boundary condition is $\tilde{q}^-(x)=\pi(x)$ for $x\in A$
and $\tilde{q}^+(x)=0$ for $x\in B$.
 \leqref{eqn:q-t}  has the same form as \leqref{eqn:q+}
by transposing the transition kernel $P$.
There are two reasons for introducing $\tilde{q}^-$: (1)
the reactive current  $J$, rather than $q^-$ itself,  is of more interest in understanding the mechanics of transition and it is not necessary to calculate $q^{-}$ explicitly in order to obtain $J$; (2) the numerical method to calculate $q^{-}$ directly is instable under small noise intensity and this problem can be resolved by calculating $\tilde{q}^-$ instead.

The system \eqref{eqn:q-t}  and \eqref{eqn:q+} together with the boundary condition
\eqref{eqn:bc} can be solved as a  linear system
after  discretizing the spatial domain
$\dom=[0,1]$.
$q^+(x)$ and $\tilde{q}^-(x)$ typically exhibit  boundary layers or discontinuities at the
boundaries of  $A$ and $B$. In our numerical discretisation,
the spatial mesh grid is adjusted in a moving mesh style to
distribute more points near the boundaries
  by checking the derivatives $|\nabla q^+|$ and $|\nabla \tilde{q}^-|$
(refer to \cite{aMAM2008} for details).

Since  \begin{equation}
J(x,y) = q^{-}(x)\pi(x)P(x,y)q^{+}(y) = \tilde{q}^-(x)q^{+}(y)P(x,y),
\label{eqn:J}
\end{equation} %Without causing any trouble, set $J(x,y) = 0$ when $x=y$.
then we can see from \eqref{eqn:q+} and \eqref{eqn:q-t}
that
\begin{equation}\label{eqn:Jx}
\int_{y\in \dom} J(x,y) dy =  q^-(x) \pi(x) \int_{y\in \dom} P(x,y)q^+(y) \d y
= \tilde{q}^-(x) ) q^+(x),  ~~ \forall x\notin (A\cup B);
\end{equation}
 \begin{equation}\label{eqn:Jy}
 \int_{x \in \dom} J(x,y) \d x
=
 q^+(y)\int_{x \in \dom}   P(x,y)   \tilde{q}^-(x) \d x=
\tilde{q}^-(y)  q^+(y), ~~ \forall y\notin (A\cup B).
\end{equation}
The above quantity on the right hand sides is
 actually  the probability density of  reactive trajectories:
\[\pi^\Rset(x) := q^-(x) \pi(x) q^+(x)=\tilde{q}^-(x)q^+(x), ~~
\forall x\in \dom \setminus (A\cup B).\]
Under ergodicity condition, this probability density $\pi^\Rset(x)$ corresponds to the following
time average:
$\pi^\Rset(x)\d x=\lim_{N\to \infty}\frac{1}{2N+1}\sum_{-N}^N \indi_\Rset(n) \indi_{[x, x+\d x)} (X_n).$
 \leqref{eqn:Jx} and  \leqref{eqn:Jy} together   show that
 \begin{equation}
 \label{eqn:J0}
 \int_{y\in \dom} J(x,y)\d y = \int_{y\in \dom} J(y,x)\d y
 =\pi^\Rset(x), ~~~\mbox{ for any }  x\in \dom\setminus (A\cup B).
 \end{equation}
So, the reactive current $J(x,y)$ defines a flow at any $x \in \dom\setminus (A\cup B)$
since the in-flow is equal to the out-flow.

\begin{remark}
For the transition kernel $P(x,y)$ based on the discrete map, it is possible that
$q^\pm (x)$ is continuous only in the open set $\dom \setminus (A\cup B)$.
The one-sided limit from  the open set $\dom \setminus (A\cup B)$
may not equal the boundary value
at  $\partial A$ or  $\partial B$ (note that $A$ and $B$ are closed set and $\partial A\subset A$,
$\partial B\subset B$).
Thus,  there may be a jump discontinuity $q^\pm (x)$ at $x\in \partial A \cup \partial B$.
 Refer to Figure \ref{fig:fbcf} in the next section for the example of \Lm.
This means that $J(x,y)$ in \eqref{eqn:J} may also have the jump
discontinuities whenever $x$ or $y$ crosses the boundaries  at $\partial A \cup \partial B$.
\end{remark}

\subsection{Transition rate and   most-probable-last-passage periodic point}
\label{ssec:mfp}

The reactive current $J$ allows us to calculate
how frequently the transition occurs from $A$ to $B$,
i.e.,   the transition rate.
The {\bf  transition rate}  is the average number of
transitions from $A$ to $B$ per unit time, defined by
\[\kappa_{AB}:=\lim_{N\to \infty} \frac{\# \{\mbox{transitions from } A \mbox{ to } B
\mbox{ in } [-N,N] \}}{2N+1}.\]
With  the definition of the set $\Rset$, we can rewrite the above as
\[ \kappa_{AB} = \lim_{N\rightarrow \infty} \frac{1}{2N+1} \sum_{-N}^{N} \indi_{A}(X_n)
\indi_{\dom\setminus A}(X_{n+1})\indi_{\Rset}(n).\]
\begin{remark}
\label{rmk:kab}
When the set $A$ is the union of disjoint compact subsets $A=\cup_{i=1}^K A_i$,
then it is obvious that the  $A$-$B$ transition rate
has the following decomposition
\[\kappa_{AB}=\sum_{i=1}^K  \kappa_{A_i B}
:=\sum_{i=1}^K \lim_{N\to \infty} \frac{1}{2N+1} \sum_{-N}^{N} \indi_{A_i}(X_n)
\indi_{\dom\setminus A}(X_{n+1})\indi_{\Rset}(n),\]
where $\Rset$ still means  the $A$-$B$ transitions.
Then  the ratio
$\displaystyle \frac{\kappa_{A_i B}} {\kappa_{AB}} $
is exactly the  probability that the reactive trajectory
selects the subset $A_i$ to leave the set $A$ during
its last stay in the set $A$.
% i.e., the probability $  \pp_x\set{X_{H^-_{AB}(0)} \in A_i}$
%averaged over  the equilibrium distribution $\pi(x)$.
\end{remark}

Using the ergodicity,     the transition rate  is calculated as follows
\begin{equation}
\label{eqn:rate}
\begin{split}
\kappa_{AB}  &= \int_{x\in A}\int_{y\in \dom\setminus A} J(x,y)  \d y  \d x
= \int_{x\in A}\int_{y\in \dom} J(x,y) \d y \d x
\\
&=\int_{x\in A}  {q}^-(x) \pi(x)   \int_{y\in \dom} P(x,y)q^+(y)   \d y  \d x
\\
&=\int_{x\in A}   \pi(x)   \int_{y\in \dom} P(x,y)q^+(y)   \d y  \d x,
\end{split}
\end{equation}
where the definition $J(x,y)=\pi(x)P(x,y)q^-(x)q^+(y)$
and the facts that $J(x,y)=0$ for $y\in A$ and
$q^{-}(x)=1$ for $x\in A$ are applied.

From \leqref{eqn:J0}, it is clear that
\[
\int_{x\in   A\cup B}  \int_{y \in \dom}J(x,y)  \d y   \d x=
\int_{x\in   A\cup B}  \int_{y \in \dom}J(y,x)  \d y   \d x .
\]
By \leqref{eqn:J=0},  the above equality becomes
\[\int_{x\in    A}  \int_{y \in \dom}J(x,y)  \d y   \d x=
\int_{x\in   B}  \int_{y \in \dom}J(y,x)  \d y   \d x .\]
Thus, there is an equivalent formula for the transition rate:
\begin{equation}\begin{split}
\kappa_{AB}
&= \int_{x\in   B}  \int_{y \in \dom}J(y,x)  \d y   \d x
=\int_{ y\in   B}  \int_{x \in \dom}J(x,y)  \d x   \d y.
\end{split}
\end{equation}

The transition rate \eqref{eqn:rate}
is the total contribution  of the reactive current out of $A$ and into $B$.
To distinguish the different points
in $A$,  where the reactive current $J$ initiates,
we    introduce the following two functions $r_{AB}^{-}(x)$
and $r_{AB}^{+}(y)$ to represent
 the local contribution of the reactive current to the reaction rate:
 \begin{align}\label{eqn:r-}
 r_{AB}^{-}(x)
& := \int_{\dom\setminus A}J(x,y)  \d y  = \int_{\dom} J(x,y) \d y, \qquad \text{for } x\in A,\\
r_{AB}^{+}(y) &:= \int_{\dom\setminus B} J(x,y) \d x = \int_{\dom} J(x,y) \d x, \qquad \text{for }y\in B.
\label{eqn:r+}
\end{align}
Note that like \leqref{eqn:rate}, $r_{AB}^{-}(x)= \pi(x)   \int_{y\in \dom} P(x,y)q^+(y)   \d y$,
requiring only the forward committor function $q^+$.

It is easy to see that $r^-_{AB}$ and $r^+_{AB}$ defined in (4.16) and (4.17), after normalization, are known as the reactive exit and reactive entrance distributions in
the transition path theory \cite{Cameron::LJcluster2014,Lu::PTRF2015}.
Indeed, by Remark \ref{rmk:kab},
the probability density function of the last passage position  on $A$
of a typical reactive trajectory   is then given by $\displaystyle \frac{r^-_{AB}(x)}{\kappa_{AB}}$
(note   $\int_A r^-_{AB}(x) \d x =  \kappa_{AB}$).
Similarly, the   probability density function  of the first entrance position on $B$
of a typical reactive trajectory  is then given by $\displaystyle \frac{r^+_{AB}(y)}{\kappa_{AB}}$.
We then  define the {\it  most-probable-last-passage point}  in $A$ as
 \begin{equation}
\hat{x} :=  \underset{x\in A} \argmax\ \frac{r^-_{AB}(x)}{\kappa_{AB}}=
\underset{x\in A} \argmax\ r_{AB}^{-}(x),
\end{equation}
and the {\it most-probable-first-hitting point} in $B$ as
\begin{equation}
\hat{y} :=\arg\max_{y\in B}\  \frac{r^+_{AB}(y)}{\kappa_{AB}}= \arg\max_{y\in B}\ r_{AB}^{+}(y).
\end{equation}
Of our particular interest  is the
most-probable-last-passage  point $\hat{x}$ in $A$.
We can think of this point as
the most  $A$-$B$ ``reactive" point in the set $A$.
In terms of instability due to the noisy perturbation,
this     point means the least stable location
in the set $A$ conditioned on the
transitions from $A$ to $B$.

For the problem of the periodic orbits $\po$ in the logistic map,
the   set $A$ is defined  as the union of neighbours of the $T$ periodic points
$\pop_1, \cdots, \pop_T$, i.e.,
$A=\underset{1\le i \le T}\cup[\pop_i-\delta_a, \pop_i+\delta_a]$.
Since $\delta_a$ is small, we can use
\begin{equation}
 r_{AB}^{-}(i) := \frac{1}{2\delta_a} \int_{\pop_i-\delta_a}^{\pop_i+\delta_a} r_{AB}^{-}(x) \d x.
\end{equation}
to represent the contributions to the total flux $\kappa_{AB}$
from the point $\pop_i$.
We define  the {\bf  most-probable-last-passage periodic point}
(abbreviated to ``MPLP")
as  the point  $\pop_{\hat{i}}$ having  the maximal value
$\{r_{AB}^{-}(i):i=1,\cdots, T\}$. This MPLP  is
 the most unstable periodic point in the sense of
transition from the periodic orbit $\po=(\pop_1,\cdots,\pop_T)$ to the set $B$.
%Our  numerical results shows that  for a proper size of width $\delta_a$
%for this point $\pop_{\hat{i}}$,
%$r^-_{AB}(x)$    is always the smallest one for all $x\in A_{\hat{i}}$.

\begin{remark}
The   transition rate
$\kappa_{AB}$ is the integration
over $x\in A$ for the function
\[ \pi(x)  \int_{y\in \dom\setminus A} P(x,y)q^+(y)   \d y  \d x.\]
As mentioned in Remark \ref{rk:pdfflux},
compared with the PDF flux defined in \leqref{eqn:Fa}
(where $\rho=\pi$),
the   difference between these two formulations  is that $q^+(y)$ is multiplied
onto $P(x,y)$ here.
The inclusion of this  forward committor function
indicates that  in the \tpt, the object of focus  is   the $A$-$B$ reactive trajectories,
which   have to reach the target set $B$ {\emph {before}}
returning   to $A$.
The trajectories counted in
the PDF flux  \eqref{eqn:Fa}
is a much larger set containing
those    trajectories  which failed to reach $B$   and
 return to $A$ again.
 So, for the same stochastic system, $\kappa_{AB}$  is usually much smaller than
 the quantity
 $\mathcal{F}^-_A$ in \leqref{eqn:Fa}
 unless $B$ is infinitely  close to $\dom\setminus A$.
% In short, the \tpt\ is a sophisticated generalization of the
% PDF flux method used in \cite{Billings::virus2002}.
\end{remark}

The   definition of the above MPLP periodic point
  is associated with the integration of the reactive probability current $J(x,y)$ for  all $y\in \dom\setminus A$.
It    does not take account of  what happens after the reactive  current leaves $A$ from the point $x$.
So,  it is possible that,  once the  reactive  current flows out of $A$ from the
MPLP point $\hat{x}$,
the reactive current could quickly diverge and  spread out,
and as a result, in terms of the transition paths from $A$ to $B$,
different transition paths can carry significantly different values of the
reactive currents.
 We then need to
    find the dominant ones among all the transition  paths connecting $A$ and $B$.
The starting points in $A$ of the dominant transition  paths
will give our second   description
of the stochastic instabilities to distinguish the periodic points
in $\po$.  Apparently,  when the dominant transition paths are not unique
due to the complexity of the problem,
it is   possible that
the starting points of these dominant transition paths
may lie in multiple subsets $ A_i$ for the case of $A=\cup_{i=1}^T A_i$, which
means that all these subsets (or the periodic points) are equally
instable  by  this path-based   criterion.

\subsection{Competency   and   maximum competency periodic point}
\label{ssec:cap}

The analysis of   pathways is built on
the   {\bf effective reactive probability current } $J^{+}(x,y)$, which is defined by
\begin{equation}
J^{+}(x,y) := \max(J(x,y)-J(y,x),0).
\end{equation}
$J^+(x,y)$ is always non-negative and represents the net  reactive flux
from $x$ to $y$.
We may write
\[
J^{+}(x,y)=\frac {J(x,y)-J(y,x)+\lvert J(x,y)-J(y,x)\rvert}2.
\]
Then using \leqref{eqn:J0}, we obtain
 \begin{equation*}
 \int_{y\in \dom} J^+(x,y)\d y = \int_{y\in \dom} J^+(y,x)\d y, ~~~\mbox{ for any }  x\in \dom\setminus (A\cup B).
 \end{equation*}
When $y\in A$, $J(x,y) = 0$ and it follows that when $x\in A$,
$J^{+}(x,y) = \max(J(x,y)-J(y,x), 0) = J(x,y)$. So, the formula of the rate
\eqref{eqn:rate} can also be written in terms of the effective current $J^+$:
\[   \kappa_{AB} = \int_{x\in A, y\in \dom} J(x,y) \d x \d y=
 \int_{x\in A, y\in \dom} J^+(x,y) \d x  \d y.
\]

The effective current $J^+(x,y)$ naturally leads to a series of concepts
about the transition paths. These concepts are well described for
countable discrete space in \cite{Metzner::mjp2009}.
Indeed, in terms of the algorithms, we can divide the continuous domain
$\dom$ into a   large number of very fine intervals (much smaller than the widths $\delta_a$
and $\delta_b$) and apply the discrete algorithms  based on the  graph theory
described  in \cite{Metzner::mjp2009}.
The theoretical formulation we give below is for a continuous space domain,
and we believe this formulation has its own interest.
To represent the functionality of the effective current $J^+$,
we shall use    a generic
two-dimensional function $f(x,y)$, which is
defined on $\dom\times \dom$, associated with the given disjoint subsets
$A$ and $B$.
This function $f(x,y)$ is an analogue of the weight
for an edge from one node $x$ to another $y$ in the graph theory.
 Clearly,  $f$ has to meet the   properties that
$J^+$ has.
We assume that the triplet $(A,B,f)$ for a compact state space $\dom$
satisfies the following assumption.
\begin{assumption}
\label{asp}
\
\begin{enumerate}
%\item $\dom$ is compact.
 \item The sets $A$ and  $B$ are disjoint nonempty closed subsets of
 the state space $\dom$
and  $A\cup B \subsetneqq \dom$;
\item $f(x,y)$ is always non-negative for all $(x,y)\in \dom \times\dom$ and
$$
f(x,y)=0,~~ \mbox{if }  x\in  B, y\in \dom \mbox{  or  }  x\in \dom,  y \in A.
$$
\item
$
f(x,x)=0, ~~ \mbox{for } x\in \dom.
$
\item
For any $x\in \dom\setminus (A\cup B)$,
\[
\int_{y\in \dom} f(x,y)\d y = \int_{y\in \dom} f(y,x)\d y.
\]
\item
$f(x,y)$ is bounded and  piecewise continuous in
$\dom\times \dom$.
\end{enumerate}
\end{assumption}

\begin{definition}
\label{def:tp}
Given two disjoint subsets $A'$, $B'$ in $\dom$ and
the triplet  $(A, B, f)$ satisfying Assumption \ref{asp},
for any $n\in \nn$,
 $\pw = (\omega_0, \omega_1, \dots, \omega_n)\in \dom \times \cdots \times \dom$ is
 called an $A'$-$B'$ {\bf  transition path}   associated with  $(A,B,f)$,
  if
 \begin{enumerate}
  %\item  $\{\omega_i:i=0,1,\cdots,n\}$ are distinct;
 \item
  $\omega_0\in A'$, $\omega_n\in B'$;
  \item $f(\omega_i, \omega_{i+1}) > 0$ for $0\le i \le n-1$.
\end{enumerate}
\end{definition}
Note that property (2) in Assumption \ref{asp} implies that $\omega_i\notin (A\cup B)$ for all $1\le i \le n-1$.

We actually use $A'=A$ (or $A'\subset A$) and $B'=B$ in most cases.
Occasionally, we need a different set $B'$ from $B$.
The following definition  of the path competency
is from the graph theory.

\begin{definition}
\label{def:capp}
We     define the {\bf competency} of a path $\pw = (\omega_0, \omega_1, \dots, \omega_n)$
as the minimal value of $f(\omega_i,\omega_{i+1})$
for all $0\leq i\leq n-1$,  that is,
\[
\Cp(\pw):=\min_{0\le i\le n-1}f(\omega_i,\omega_{i+1}).
\]
\end{definition}
\begin{remark}
The notion of ``competency'' defined above is referred to as capacity in the context of the graph theory.
However, the terminology ``capacity" is also used and plays a significant role in the classical potential theory for stochastic systems which is
closely related to the transition path theory. So to avoid confusion, we adopt a different terminology ``competency".
\end{remark}
Property (2) in Definition \ref{def:tp} implies that the competency of any $A'$-$B'$ transition path is always strictly positive.

\begin{definition}
With the same assumption in Definition \ref{def:tp},
a subset $\mathcal{C}$ of the product space $\dom \times \dom $  is called {\bf $A'$-$B'$ $f$-connected},
 if there exists at least one
$A'$-$B'$ transition path $\pw = (\omega_0, \omega_1, \dots, \omega_n)$
  for some $n\geq 1$, associated with the triplet  $(A,B,f)$,  such that
every directed  edge  $(\omega_i, \omega_{i+1})$ belongs to $ \mathcal{C}$ for all $0\le i\le n-1$.

The  collection of all  $A$-$B$ transition paths  with length $n$ and all edges contained  in the
 set  $\mathcal{C}$
is denoted by
$\mathbb{G}_n(\mathcal{C})$.
$\mathbb{G}(\mathcal{C}):=\cup_n \mathbb{G}_n(\mathcal{C})$.

\end{definition}

 We   drop out the function $f$ most of  the time and simply say the set $\mathcal{C}$
 is $A'$-$B'$  connected.
We are particularly interested in the special set  $\mathcal{C}$
  in the form of the super level set of the function $f$.

\begin{definition}
\label{def:cpf}
With the same assumption in Definition \ref{def:tp}, define the superlevel  set
of the function $f$ for any non-negative  real  number $z$, $$\Lr_z := \{(x,y)\in \dom \times \dom: f(x,y) \ge z \}.$$
The  {\bf  $A'$-$B'$  competency} of the function $f$,  denoted as $z^*(A',B')$,  is defined as
\begin{equation}
\label{eqn:z*}
z^{*}(A',B') := \sup\set{z\geq 0: \Lr_z \text{  is } A'\mbox{-}B' \text { connected}}.
\end{equation}
$\Lr_{z^{*}(A',B')}$ is call the {\bf minimal $A'$-$B'$ connected superlevel set} of $f$
if the maximizer can be reached:
\[
z^{*}(A',B')= \max\set{z\geq 0: \Lr_z \text{  is } A'\mbox{-}B' \text { connected}}.
\]
As a convention, when $A'$ and $B'$ are not specified,
$A'$ is $A$ and $B'$ is $B$ by default and we simply say the competency of the function $f$,
 the minimal connected set and denote $z^*(A',B')$ as $z^*$.

\end{definition}

\begin{remark}
\label{rmk:mcss}
The relation between Definition \ref{def:capp} and   Definition \ref{def:cpf}
is that
\begin{equation}
\label{eqn:cap}
z^*(A',B')=\sup\set{\Cp(\pw): \pw \text{  is an } A'\mbox{-}B' \text { transition path}}.
\end{equation}
Indeed, if $\pw$ is an $A'$-$B'$ transition path, then $L_{z}$ is $A'$-$B'$ connected for $z\le\Cp(\pw)$;
and conversely, if $L_{z}$ is $A'$-$B'$ connected, then any $A'$-$B'$ transition path $\pw=(\omega_0,\omega_1,\cdots,\omega_n)$ with edges contained in $L_z$
must satisfy $f(\omega_i,\omega_{i+1})\ge z$ for all $0\leq i\leq n-1$,
thus $\Cp(\pw)\ge z$.
In particular,   all the $A$-$B$ transition paths with edges
in the minimal $A$-$B$ connected superlevel  set $\Lr_{z^*}$
must have the same competency $z^*$ as the function $f$.
%The reason is as follows:
%If there exits one path $\pw\in \mathbb{G}(\Lr_{z^*})$ with $\Cp(\pw)>z^*$, then
%the competency of $f$ will be at least $\Cp(\pw)$, not equal to $z^*$;
%Since $L_{z^*}$ is $A$-$B$ connected,
%any $\pw=(\omega_0,\omega_1,\cdots,\omega_n)\in \mathbb{G}(\Lr_{z^*})$ must satisfy $f(\omega_i,\omega_{i+1})\ge z^*$ for all $0\leq i\leq n-1$,
%thus $\Cp(\pw)\ge z^*$.
\end{remark}

\begin{definition} \label{def:dtp}
With the same assumption in Definition \ref{def:tp},
let $z^*(A',B')$ be the $A'$-$B'$ competency of $f$
in Defintion \ref{def:cpf}, if $\Lr_{z^*(A',B')}$ is $A'$-$B'$ $f$-connected,
we then call all the $A'$-$B'$ transition paths with edges in $\Lr_{z^*(A',B')}$
  the  {\bf $A'$-$B'$ dominant transition paths}.
The $A$-$B$ dominant transition paths
 are simply called   the  dominant transition paths.
\end{definition}

 In   our problem about the periodic orbit $\po=(\pop_i)_{i=1,\cdots,T}$,
 the set $A$ is $\cup_{i=1}^T A_i $ where  $A_i= [\pop_i-\delta_a, \pop_i+\delta_a]$.
 Note that the following important fact from \eqref{eqn:cap},
 \[z^*(A,B)=\max_i z^*(A_i,B).\] Therefore,
 we propose to make use of the capacities  $z^*(A_i,B)$ for $1\le i \le T$
 to compare the instability of each $\pop_i$.
The point  $\pop_{\hat{i}}$  such that $z^*(A_{\hat{i}},B)=\max_i z^*(A_i,B) $ is defined
as
the {\bf  maximum competency periodic point} (MCPP).
The interpretation of this MCPP is that
there exists a transition path emitting from this MCPP (more precisely, its window $A_{\hat{i}}$)
whose competency is larger than any transition path emitting from any other periodic point.
Thus this MCPP is deemed as the most active (least stable)  periodic point
in the noise-induced transition from $A$ to $B$.
If this MCPP $\pop_{\hat{i}}$ is unique, then all the dominant transition paths
will start from $A_{\hat{i}}$.
In case that  the maximizers  are not unique,
the capacities $z^*(A_i, B)$ still can in general  give a rank in terms of stochastic instability for all periodic points $\po=(\pop_i)$.

In the community  of graph algorithms and  network optimization,
the dominant transition path   is called   the
widest path, also known as the
bottleneck shortest path   or the maximum competency path.
There are plenty of practical  algorithms to find the widest path \cite{Ahuja1993NFT}.
In what follows, we discuss the identification
of the $A$-$B$ competency  $z^*$ and the dominant transition paths.
The motivation here is not to present the details of the practical implements for discrete state space,
but to demonstrate the concepts and  the related theoretical properties in the continuos space.

It is easily seen from \eqref{eqn:cap} that $z^*>0$.
On the other hand, for $z>\sup_{\dom \times \dom} f$,  $L_z$ is empty.
So,    the  competency $z^*$ of  $f$
satisfies $0 < z^{*}  \le \sup_{\dom \times \dom} f < \infty$.
The following properties are obvious: (1)
  If $ \Lr_{z_1}$ is   connected, then so is  $\Lr_{z_2}$
for any $z_2<z_1$;
(2) $\Lr_z$ is     connected  for any $0 < z < z^{*}$;
(3)  $\Lr_z$ is not connected for any  $z>z^*$.
So, one can use a binary search algorithm to compute   the competency $z^*$ of $f$
 within the interval $(0 , \sup_{(x,y)\in\dom\times\dom} \, f(x,y) ]$.
 Then the numerical result for  $z^*$ is   a tiny interval $[z^*_l, z^*_u]$
 bracketing the true value $z^*$.
To judge a given set
$\Lr_z$  is   $A$-$B$ $f$-connected or not,
we can use the following set-to-set map $\Phi_z$ to propagate the set $A$ until reach $B$ if it is reachable.
 The map $\Phi_z$  provides a set-tracking  algorithm
 to search the transition path from $A$ to $B$. The idea is   the analogue
 of the breadth-first  search algorithm.
The same procedure is used  to test  every $A_i$-$B$ $f$-connection
in order to identify  $z^*(A_i,B)$. Actually, since $A=\cup A_i$,
the set-tracking is performed in parallel for all $A_i$.

\begin{definition}
For any $z>0$, we can define the map $\Phi_z$ on  the collection of all subsets of $\dom$ by
$$
\Phi_{z}(C) = :\cup_{x\in C} \  \{y:  (x,y)\in L_z\},\qquad \forall C \subset \dom.
$$
Denote the compound mapping by
$$
\Phi^m_{z}(C) := \Phi_z(\Phi^{m-1}_z(C))
$$
and  $\Phi_{z}^{0}(C) := C$ by default.
\end{definition}

Let
\[ N(z) := \min\{n\geq 1: \Phi^{n}_{z}(A) \cap B \neq \emptyset \} \]
be the minimal length of the $A$-$B$ transition paths in $\mathbb{G}(\Lr_z)$,
then $N(z)<\infty$ if and only if $\Lr_z$ is   $A$-$B$ connected.

To avoid the  technicality and  ease the presentation, we  theoretically  assume  that
$\Lr_{z^*}$ is  $A$-$B$ $f$-connected, i.e.,
$z^*$ is  the maximizer in \eqref{eqn:z*}.
Numerically,
 we   check for $z$ slightly below  the numerical value $z^*_l$,
 and if for all these $z$'s, they share  exactly the same $N(z)$
 and the set $\Phi^{N(z)}_{z}(A)\cap B$ converges as $z$ approaches $z^*_l$,
 then we are able to
use the obtained numerical value $z^*_l$
  as the competency of $f$ defined in \eqref{eqn:z*}.

\subsection{Dominant transition path and dynamical bottleneck}
Calculating
 the $A_i$-$B$ competency, $z^*(A_i,B)$, suffices for
 quantifying the stochastic instabilities of    the periodic points.
 In the following last part of this section, we further discuss some additional issues about
 finding  the  $A$-$B$ dominant  transition paths
 since such paths can give us more details and insights
 of the transition mechanism, especially how the
 periodic points compete in winning the global competency $z^*$.

% We restrict   to those dominant  transition paths with  minimal path lengths
% (i.e., up to length $N(z)$)
% to exclude the possible existence  of loop.
First we  define  a pull back operation. %from a non-empty $ \Phi_{z}^{N(z)}(A)\cap B$.
\begin{definition}
Given $z \leq z^*$ and $n\geq N(z)$,  let   \[W_{z}^{n,n} := \Phi_{z}^{n}(A)\cap B,\]
if this set is nonempty. And
   for $0\le i < n$, define recursively,
\[
\begin{split}
W^{n,i}_{z} &:= \{x\in \Phi^{i}_{z}(A):\Phi_{z}(\{x\})\cap W^{n,i+1}_{z}\neq \emptyset \} .
\end{split}
\]
\end{definition}

For  any $\pw=(\omega_0,\omega_1,\cdots,\omega_n)$,
define   the canonical  projection  $\pi_i: \pw \mapsto \omega_i$.
Then we have the following property about the above set $W^{n,i}_z$.
\begin{proposition}
For any $z\in (0,z^*]$, $n\geq N(z)$,   and $ 0 \leq i \leq n$,
then
\[  W^{n,i}_z = \pi_i  \left( \mathbb{G}_{n}(\Lr_z)\right),
~~~\]
which  is to say
\begin{enumerate}
\item for any $\alpha \in W^{n,i}_z$,
there exits a transition path $\pw
=(\omega_0,\omega_1,\cdots,\omega_n) \in \mathbb{G}_n(\Lr_{z})$ with length $n$
and $\omega_i = \alpha$.
\item for any $\pw = (\omega_0,\omega_1,\cdots, \omega_{n}) \in \mathbb{G}_{n}(\Lr_z) $,
$\omega_i\in W^{n,i}_z$ for all $0\le i\le n$.
\end{enumerate}
\label{PSI}
\end{proposition}

\begin{proof}
(1):
Pick up an arbitrary   $\alpha$ in $W^{n,i}_z$,   let $\omega_i:=\alpha$,
then there exists a point, denoted as $\omega_{i+1}$, in both $\Phi_z(\{\omega_i\})$ and $W^{n,i+1}_z$.
Since $\omega_{i+1} \in W_{z}^{n,i+1}$,   we can inductively  find $\omega_j\in W^{n,j}_{z}\cap \Phi_z(\{\omega_{j-1}\})$ for $i< j\le n$; in particular, $\omega_{n} \in W^{n,n}_{z} \subset B$.
Meanwhile, since $\omega_i\in \Phi^i_{z}(A)$, then there exists
an $\omega_{i-1}$   such that $\omega_{i-1} \in \Phi^{i-1}_{z}(A)$
and $\omega_i \in \Phi_{z}(\{\omega_{i-1}\}) $. From  $\omega_{i-1} \in \Phi^{i-1}_{z}(A)$,
we similarly have $\omega_j \in \Phi^{j}_{z}(A)$ and $\omega_{j+1} \in \Phi_{z}(\{\omega_{j}\})$ for $0\le j< i$; in particular,  $\omega_{0} \in \Phi^{0}_{z}(A) = A$.
Then
 $\pw := (\omega_0,\cdots, \omega_i, \cdots, \omega_{n})$
is the desired transition path.

(2):
Let  $\pw = (\omega_0, \omega_1, \cdots, \omega_{n})$
be a transition path in the    set $\Lr_z$.
Then $\omega_0\in A=\Phi^0_z(A)$.
Note that $f(\omega_i,\omega_{i+1}) \ge z$ for all $0\leq i < n$,
then $\omega_{i+1} \in \Phi_{z}(\{\omega_i\})$.
In particular, $\omega_1\in \Phi_z(\{\omega_0\}) \subset \Phi^{1}_{z}(A)$, and
inductively,  we have  $\omega_i \in \Phi^{i}_{z}(A)$ for $0\le i \le n$. Since $\omega_{n} \in B$, thus we have $\omega_{n} \in W^{n,n}_{z}$.
Then by induction, we obtain from the definition of $W^{n,i}_z$
 that $\omega_i \in W^{n,i}_{z}$ for $0\le i\le n$.
\end{proof}

\begin{definition}
\label{def:btn}
 A  pair
$(x,y)\in \dom \times \dom$ is called
an $A$-$B$ {\bf dynamical bottleneck}, or {\bf dynamical bottleneck} for abbreviation, if
$f(x,y)=z^*$ and
$(x,y)\in W^{n,i}_{z^*} \times W^{n,i+1}_{z^*}$ for
some $n\ge N(z^*)$ and $0 \leq i < n$.
\end{definition}

\begin{proposition}
\begin{enumerate}
\item
If $(x,y)$ is a  dynamical bottleneck, then
there exists a dominant transition  path
$\pw=(\omega_0,\cdots,\omega_{n})$ in $\mathbb{G}(\Lr_{z^*})$, such that
$x= \omega_i$ and $y=\omega_{i+1}$ for some $0 \leq i < n$.
\item
If for  the given set $A$, $B$ and the function $f$,
the bottleneck is unique,  then
every dominant transition path contains
the bottleneck as one of its edges.
\end{enumerate}
\begin{proof}
(1)
From the proof of  Proposition \ref{PSI},
we see that if $x\in W^{n,i}_{z^*}$,
there must exist an $A$-$\set{x}$ transition path
$(\omega_0,\omega_1,\cdots,\omega_i=x)$ with edges in $L_{z^*}$,  and if  $y\in W^{n,i+1}_{z^*}$,
there should be a $\set{y}$-$B$ transition path $(\omega_{i+1}=y,\omega_{i+2},\cdots, \omega_{n})$
with edges in $L_{z^*}$.
Since $(x,y)\in \Lr_{z^*}$, then
putting together the above two pieces,
we obtain $\pw=(\omega_0,\omega_1,\cdots,\omega_i=x,
\omega_{i+1}=y,\omega_{i+2},\cdots, \omega_{n} )$
is a dominant transition path.

(2)
In view of Remark \ref{rmk:mcss}, for every dominant transition path $\pw = (\omega_0, \omega_1, \cdots, \omega_n)$  in $\mathbb{G}(\Lr_{z^*})$,
we have  $\Cp(\pw)=\min_i  f(\omega_i,\omega_{i+1})=z^*$.
Let $i^*=\arg\min_i  f(\omega_i,\omega_{i+1})$, then $f(\omega_{i^*},\omega_{{i^*}+1})=z^*$.
On the other hand, it follows from Proposition  \ref{PSI} that $\omega_{i^*}\in W^{n,i^*}_{z^*}$
and $\omega_{i^*+1}\in W^{n,i^*+1}_{z^*}$.
Hence $(\omega_{i^*}, \omega_{i^*+1})$ is a bottleneck by definition.
Since the bottleneck is unique,
$(\omega_{i^*}, \omega_{i^*+1})$ must be the bottleneck $(x,y)$.
\end{proof}
\end{proposition}

For the situations that the $A$-$B$  dynamical bottleneck
is   unique, which is denoted as $\Btn(A,B)=(\Btn^-(A,B),\Btn^+(A,B))$,
  we can  furthermore recursively  investigate
how the  dominant transition  paths
leave the set $A$ and reach the bottleneck $\mathbb{B}(A,B)$.
For example, we can define the bottleneck
$\Btn(A, \Btn^-(A,B))$
for the transition from $A$ to  $\Btn^-(A,B)$, i.e.,
taking  $\Btn^-(A,B)$ as $B'$.
If this bottleneck is also unique, we can continue to trace
the nested bottlenecks $\Btn(A, \Btn^-(A, \cdots))$  back
to some point in the set $A$.
The final point obtained in this recursive way
in the set $A$ is just
the MCPP we defined before.

\subsection{Comments on two criteria  of MPLP and MCPP}
\label{ssec:two}
It is normal that our two criteria
in Section \ref{ssec:mfp} and Section \ref{ssec:cap}
can give rise to different results in describing the stochastic
instabilities of  the same periodic point
in regard of different criteria  used.
The first criterion of looking for MPLP is to compare the
total out-flow of the reactive current from a periodic point.
 The second criterion of looking for MCPP is
to compare the   competency of the ``pipelines"  from a periodic point
in
transporting the reactive current to the destination $B$.
So,  it is quite reasonable that
the total flow is huge but the competency of each individual pipeline
is actually small, or the vice versa.  In a nutshell, the MPLP is for the collective behavior
of all pipelines while the MCPP is about
where the pipeline with the widest bottleneck lies.

 \section{Application to the random logistic map}
\label{sec:num}
We are now in the position to apply the above method
based on the TPT to the logistic map
for the set $A$ and $B$ specified  in Section \ref{sec:log}.
The first result is for a fixed value $\alpha=3.2$, at which
a stable period-2 orbit exists.
We shall show the numerical values of the  $A$-$B$ reactive probability current
 $J$ and the analysis of the MPLP,  MCPP and dominant transition paths.
Then,  by changing  various parameter
$\alpha$ and the noise amplitude $\sigma$,
we study how these quantities change
to affect   the individual points in one periodic orbit.
  During the discussion, we also show
some validation work
for the consistence with the  direct simulation
and the robustness with respect to $\delta_a$ and $\delta_b$.

\subsection{Results for the period-2 case}

\subsubsection{Basic quantities}

 {\it (1) invariant measure $\pi$:}
 Pick up  $\alpha = 3.2$ as an example first.
The stable period-2 orbit in this case is  $\po=(\pop_1,\pop_2) = (0.5130,0.7995)$.
The invariant measure $\pi$ at  $\sigma = 0.04$ is shown in Figure \ref{fig:pi2},
where the two peaks correspond to the
locations of  $\pop_1 $ and $\pop_2$.
It is seen that $\pi(\pop_1)<\pi(\pop_2)$, which implies that
the periodic point $\pop_2$ on the right  has higher probability  at equilibrium.
The same result  $\pi(\pop_1)<\pi(\pop_2)$ for the two periodic points $\pop_1<\pop_2$
is observed for all values of $\alpha$ between [3.02, 3.4].
Actually, when $\alpha$ increases in this interval, so does
the ratio $\pi(\pop_2)/\pi(\pop_1)$.

Figure \ref{fig:pi3} shows the invariant measure for a period-3 example
at $\alpha=3.83$. The period-3 orbit is  $\po=(\pop_1,\pop_2,\pop_3)=(0.1561, 0.5047, 0.9574)$.
To show the three peaks for this periodic orbit, a smaller $\sigma=0.008$ is set.
It is shown here that  the peak at $\pop_3=0.9574$ is   dominantly large.

\begin{figure}[htbp]
	\begin{subfigure}[b]{.45\textwidth}
	
	\includegraphics[width=\textwidth]{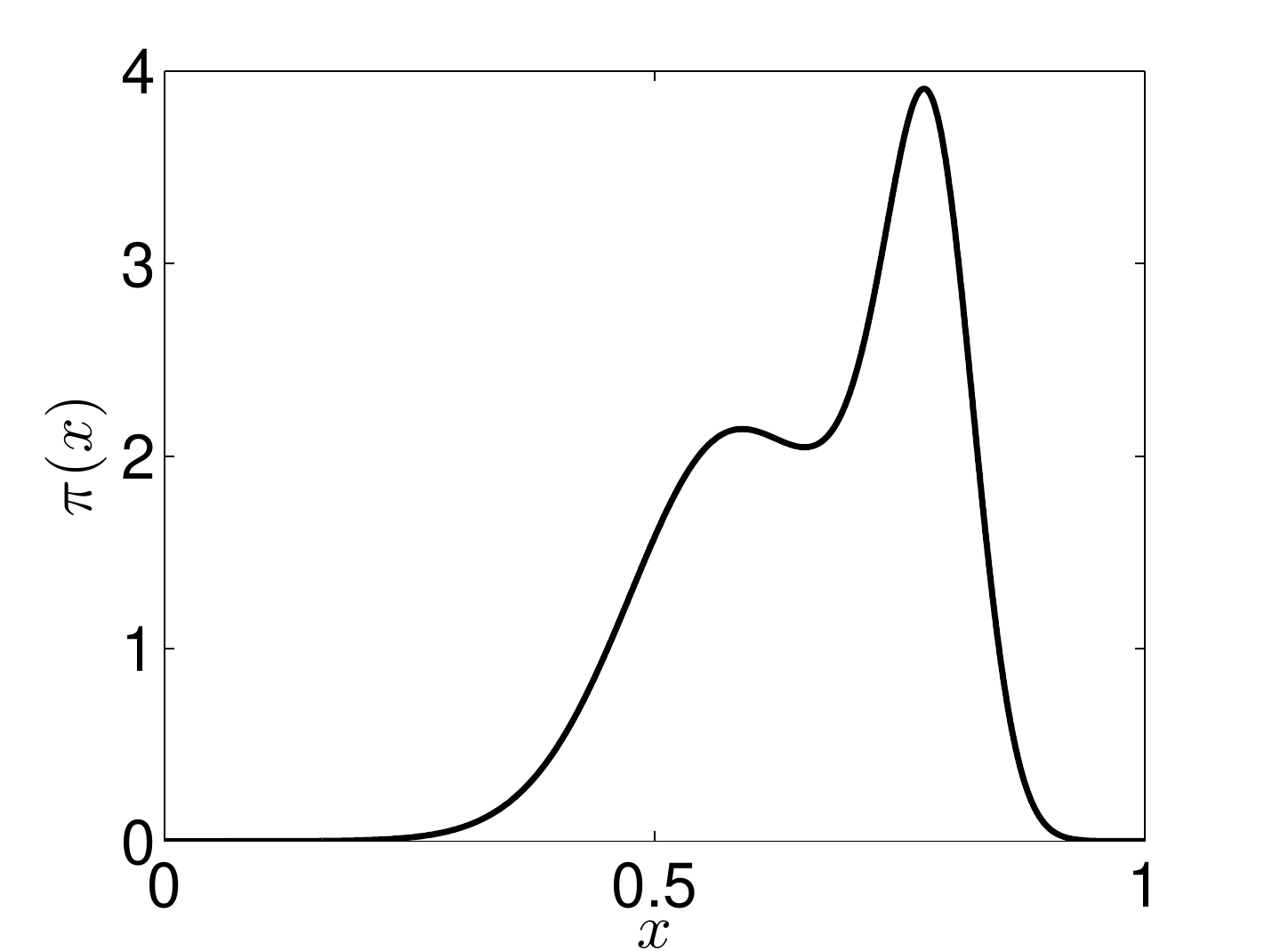}
	\caption{   $\pi(x)$ at a  period-2 case.}
	\label{fig:pi2}
	\end{subfigure}
	~
	\begin{subfigure}[b]{0.45\textwidth}
		\includegraphics[width=\textwidth]{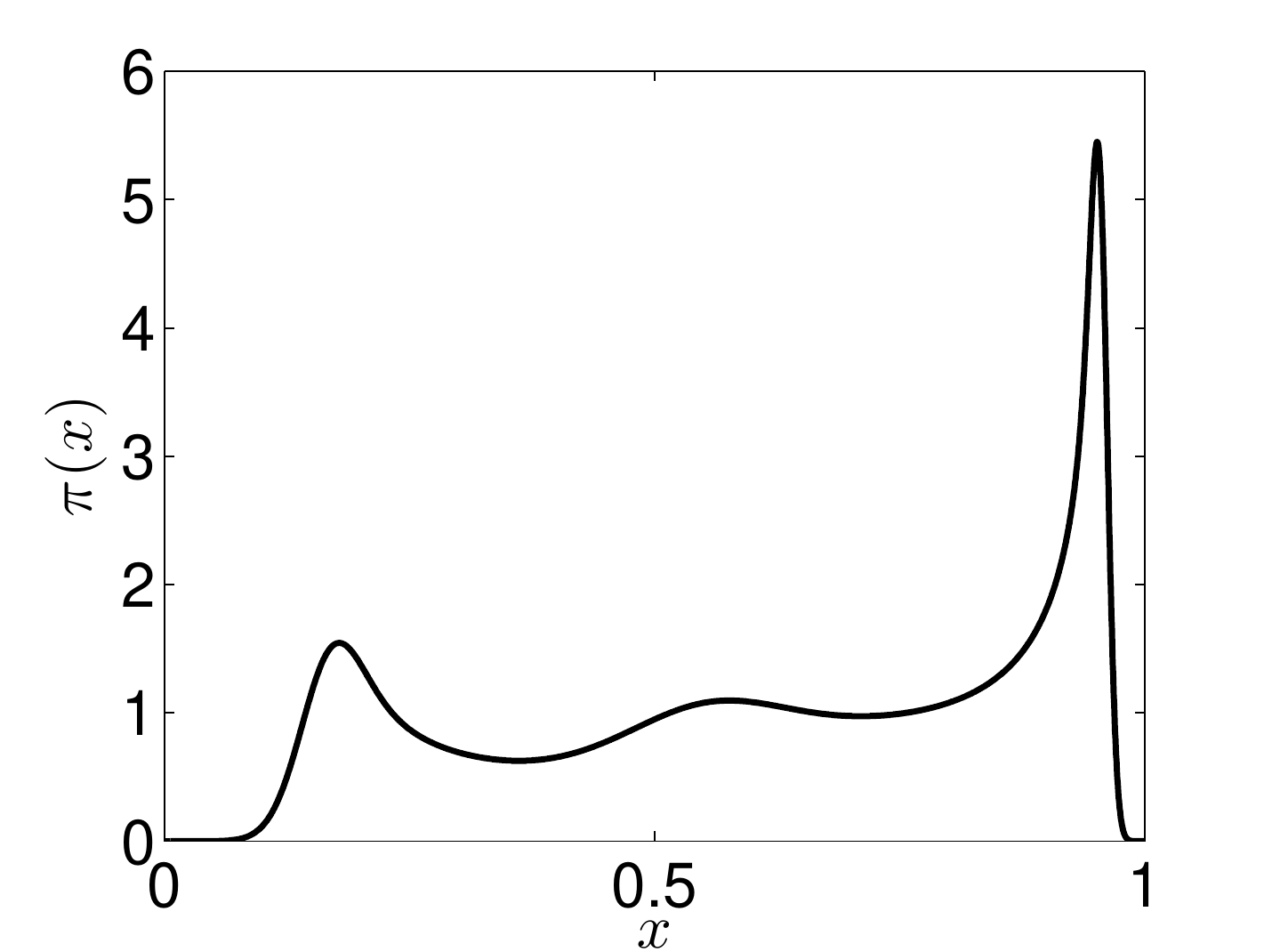}
	 	\caption{  $\pi(x)$ at a period-3 case.}
		\label{fig:pi3}
	\end{subfigure}
	\caption{The invariant probability density $\pi(x)$ for (A) period-2 case and (B) period-3 case. The parameters are (A) $\alpha = 3.2$, $\sigma = 0.04$, (B) $\alpha = 3.83$, $\sigma = 0.008$.}
	\label{fig:pi}
\end{figure}

\begin{figure}[h!]
	\begin{subfigure}[b]{0.48\textwidth}
		\includegraphics[width=\textwidth]{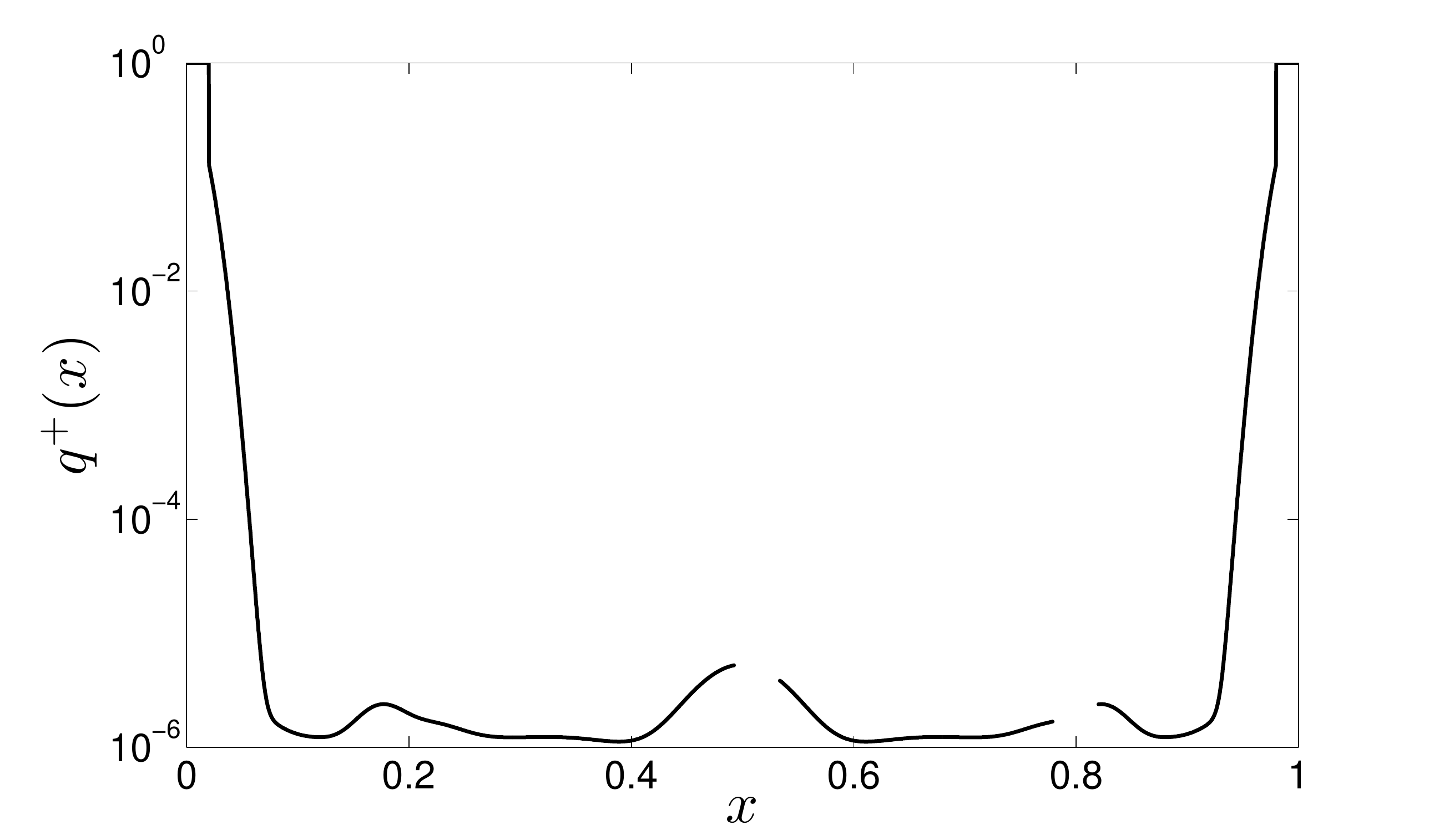}
		\caption{$  q^+(x)$ from solving \eqref{eqn:q+}. }
\label{fig:fbcf-A}
	\end{subfigure}
	\begin{subfigure}[b]{0.48\textwidth}
		\includegraphics[width=\textwidth]{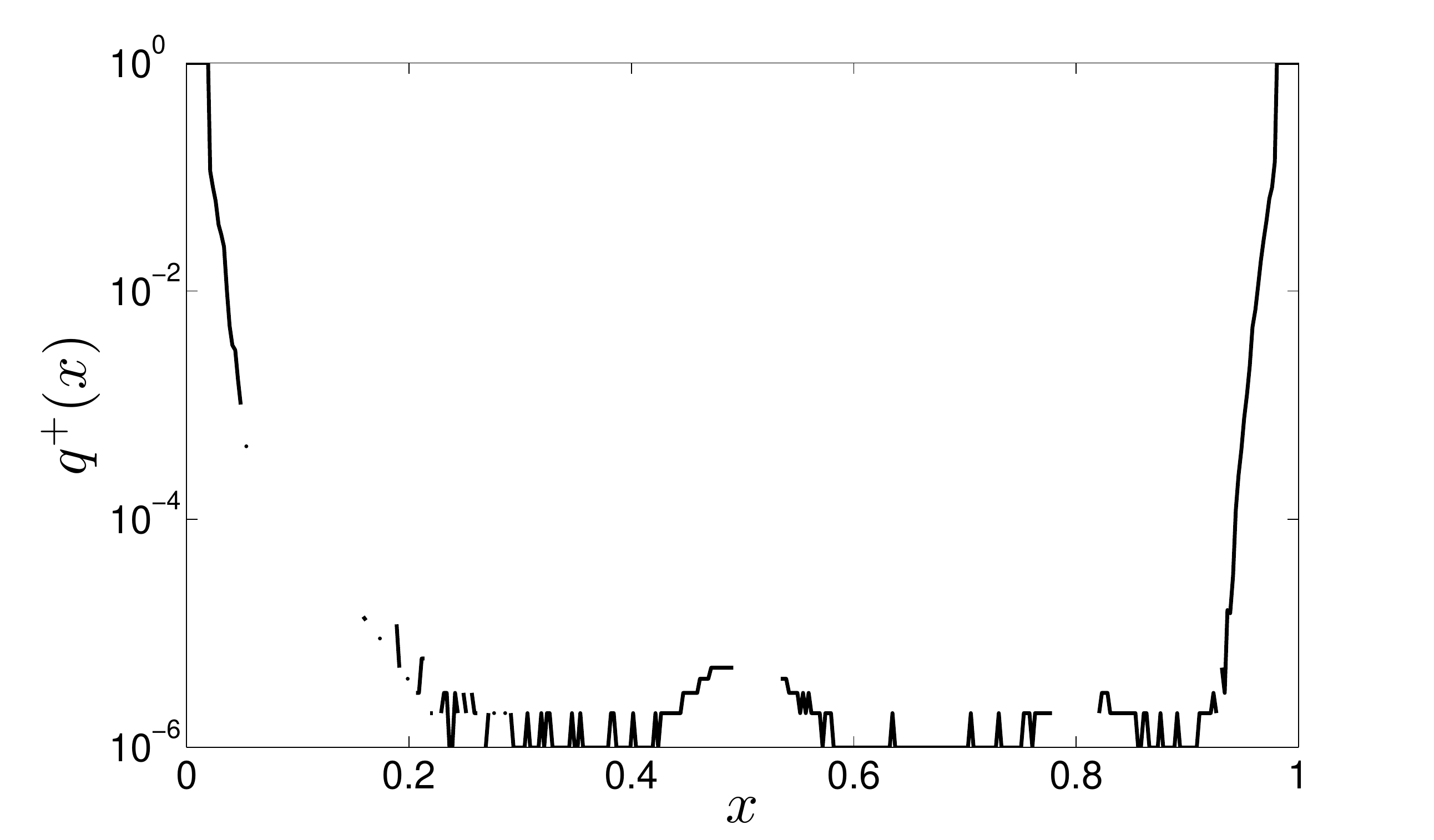}
		\caption{$ q^+(x)$ from direct simulation.}	
		\label{fig:fbcf-B}
	\end{subfigure}
	 	\begin{subfigure}[b]{0.48\textwidth}
		\includegraphics[width=\textwidth]{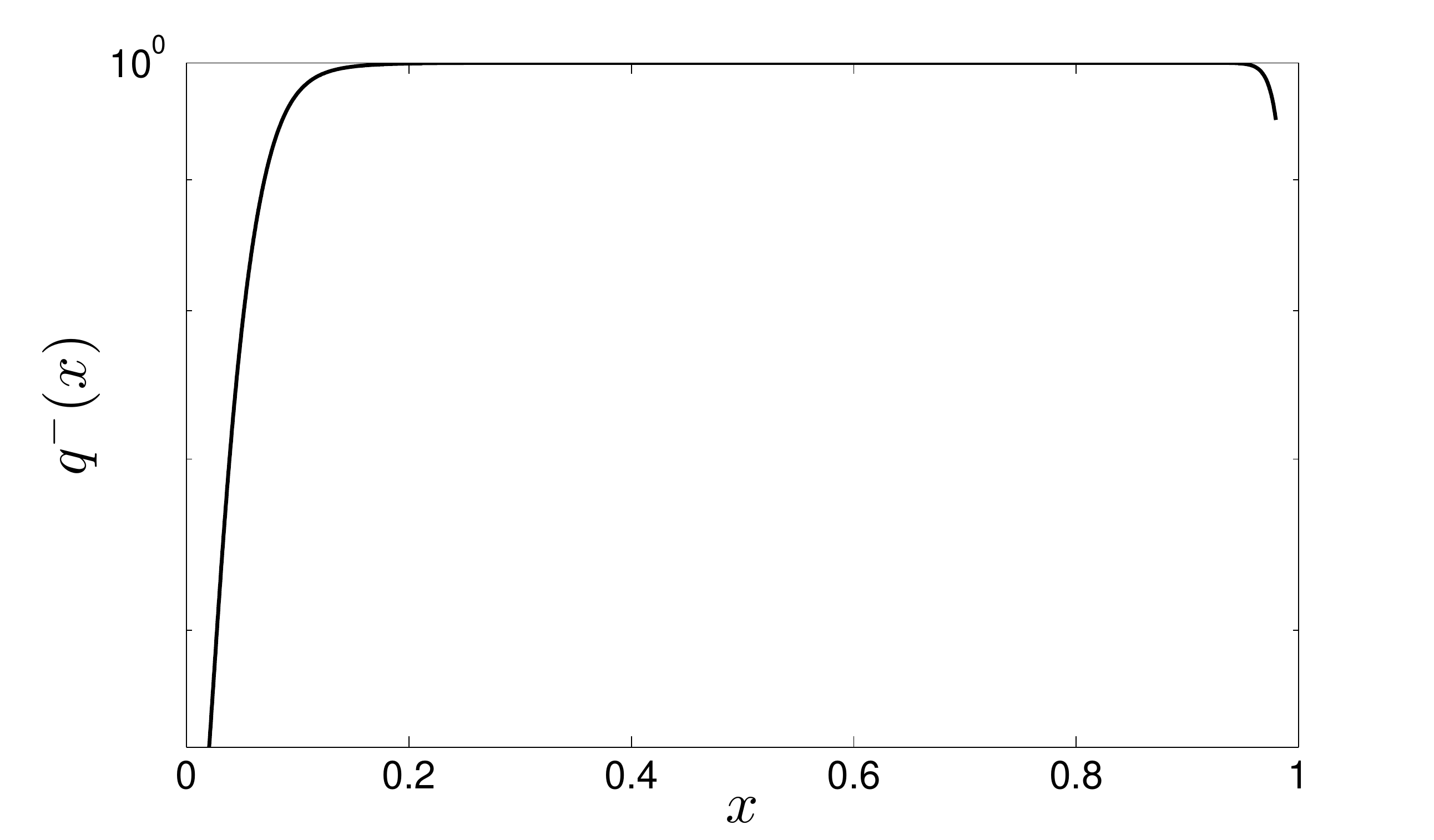}
		\caption{$   q^-(x)$ from solving \eqref{eqn:q-}.}
		\label{fig:fbcf-C}
	\end{subfigure}
	\begin{subfigure}[b]{0.48\textwidth}
		\includegraphics[width=\textwidth]{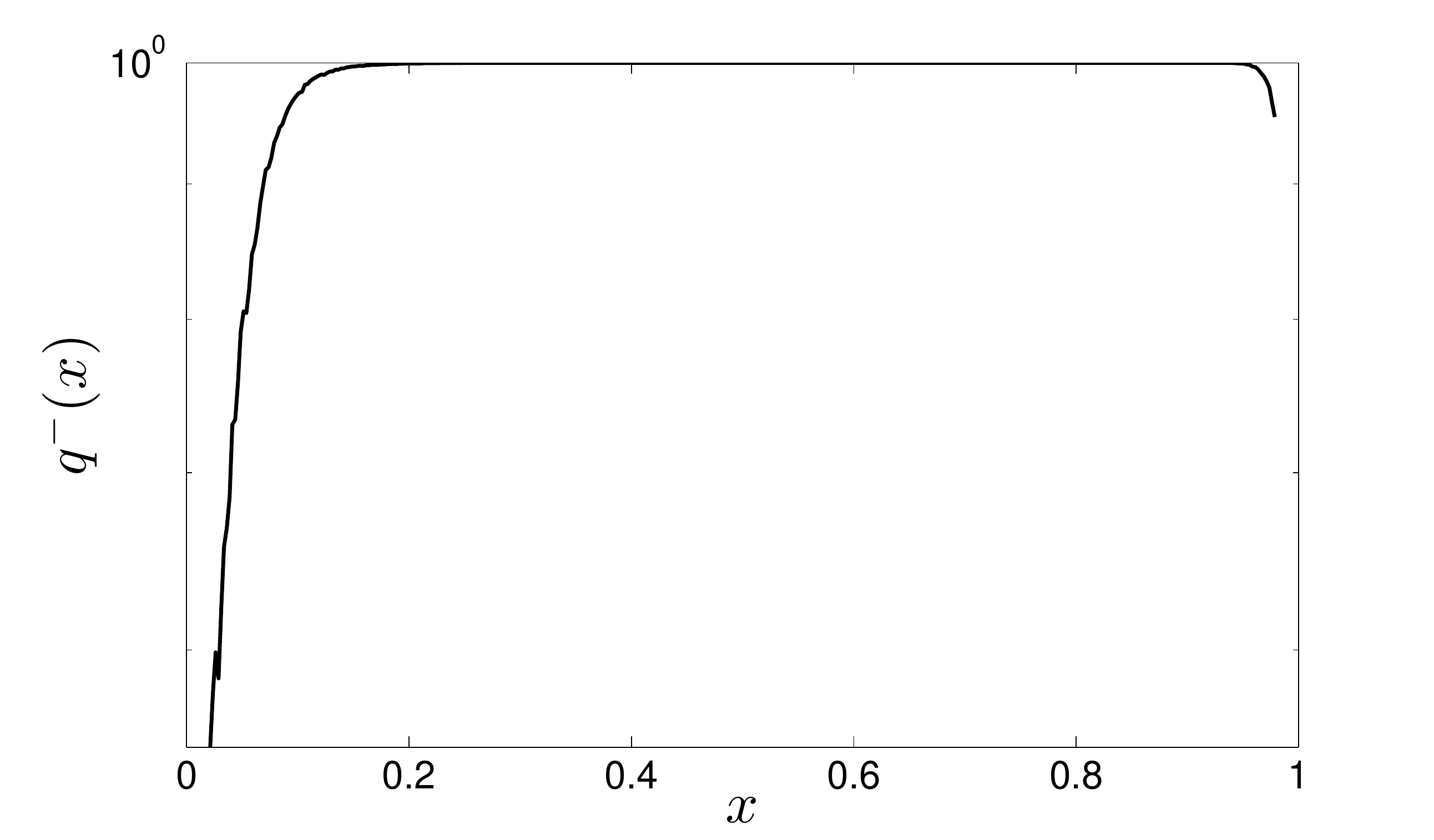}
		\caption{ $  q^-(x) $ from direct simulation. }
		\label{fig:fbcf-D}
	\end{subfigure}
	\caption{ The logarithmic plots of
	the forward committor function  (A, B) and backward committor function (C, D).   The parameters are $\alpha = 3.2$, $\sigma = 0.04$, $\delta_a = \delta_b = 0.02$.}
	\label{fig:fbcf}
\end{figure}

{\it (2) Committor  functions.}
 Choose the sets $A$ and $B$ as in \eqref{setA} and \eqref{setB}
 with  $\delta_a = \delta_b = 0.02$.
 The forward committor function  $q^+$  and backward committor function
$q^-$ at    $\sigma = 0.04$ ($\alpha=3.2$)
 are plotted  in Figure \ref{fig:fbcf}
at  the logarithmic scale.
As a comparison to    the solutions
obtained from the  finite difference scheme for   \leqref{eqn:q+} and \leqref{eqn:q-}
with $10^4$ grid size, shown  in the subplot Figure \ref{fig:fbcf-A}   and \ref{fig:fbcf-C},
the  same  committor functions
in Figure \ref{fig:fbcf-B}   and \ref{fig:fbcf-D}
are calculated
 from the statistical average
 of  a long trajectory by brute-force simulation of
 the random logistic mapping.
 The total simulation time step  is $2\times 10^{10}$
 (i.e., $N=10^{10}$ in \leqref{eqn:rate}), during which the number of successful transitions
 from $A$ to $B$ is $12238$.
 Thus the transition rate obtained from direct simulation is $6.119\times 10^{-7}$.
 The transition rate calculated from the equation \eqref{eqn:rate} is $6.008\times 10^{-7}$.

It should be emphasized  that the committor functions are not continuous
at the boundary of the sets $A $ and $B$.
The forward committor function does not even change monotonically
from $1$ to $0$. These special features
come from the nature of the discrete-time dynamical system.

\begin{figure}[htbp]
	\begin{subfigure}[b]{0.49\textwidth}
		\includegraphics[width=\textwidth]{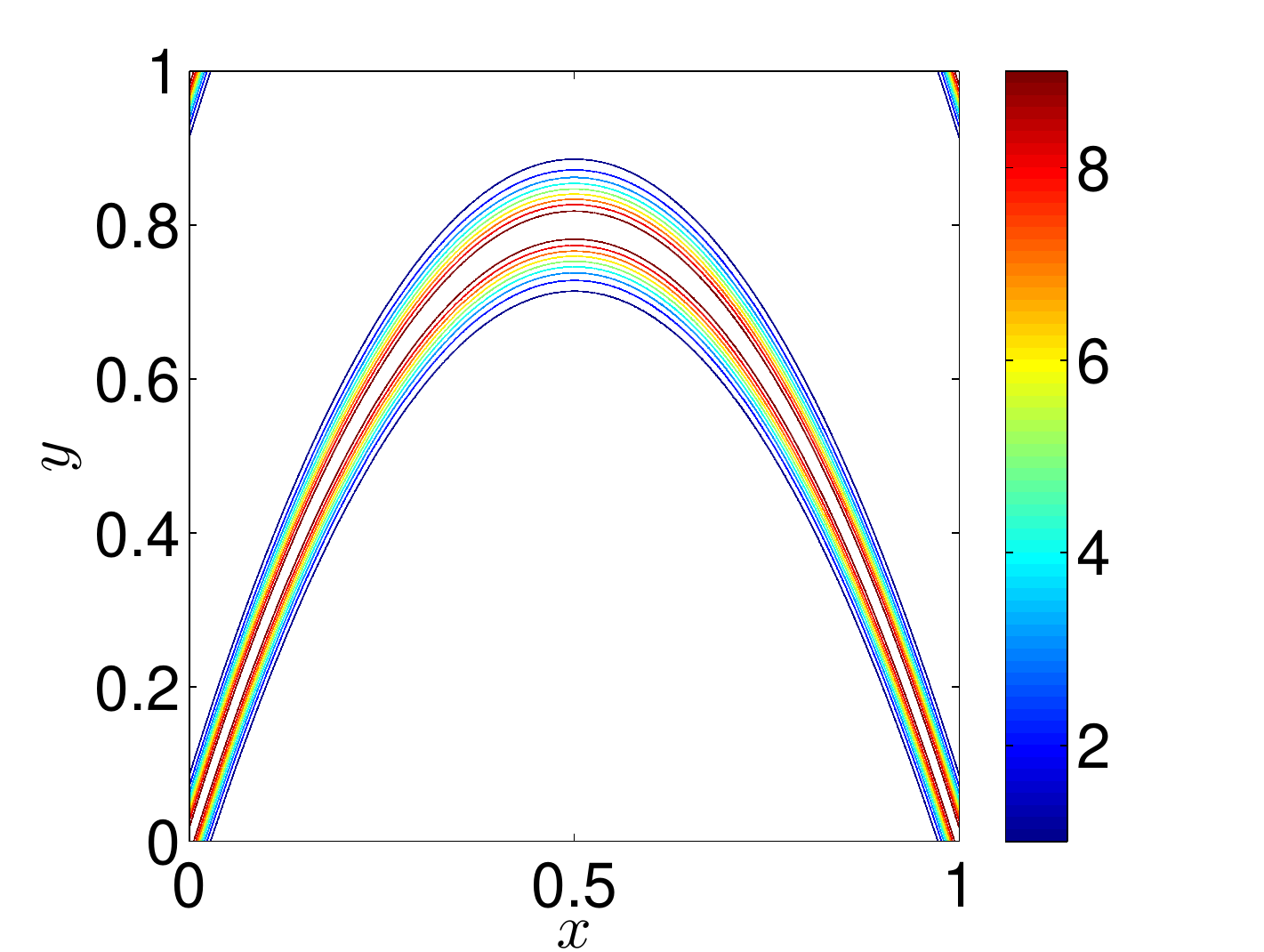}
		\caption{ $P(x,y)$.}\label{fig:P-pi-J-A}
	\end{subfigure}
	~
	\begin{subfigure}[b]{0.49\textwidth}
		\includegraphics[width=\textwidth]{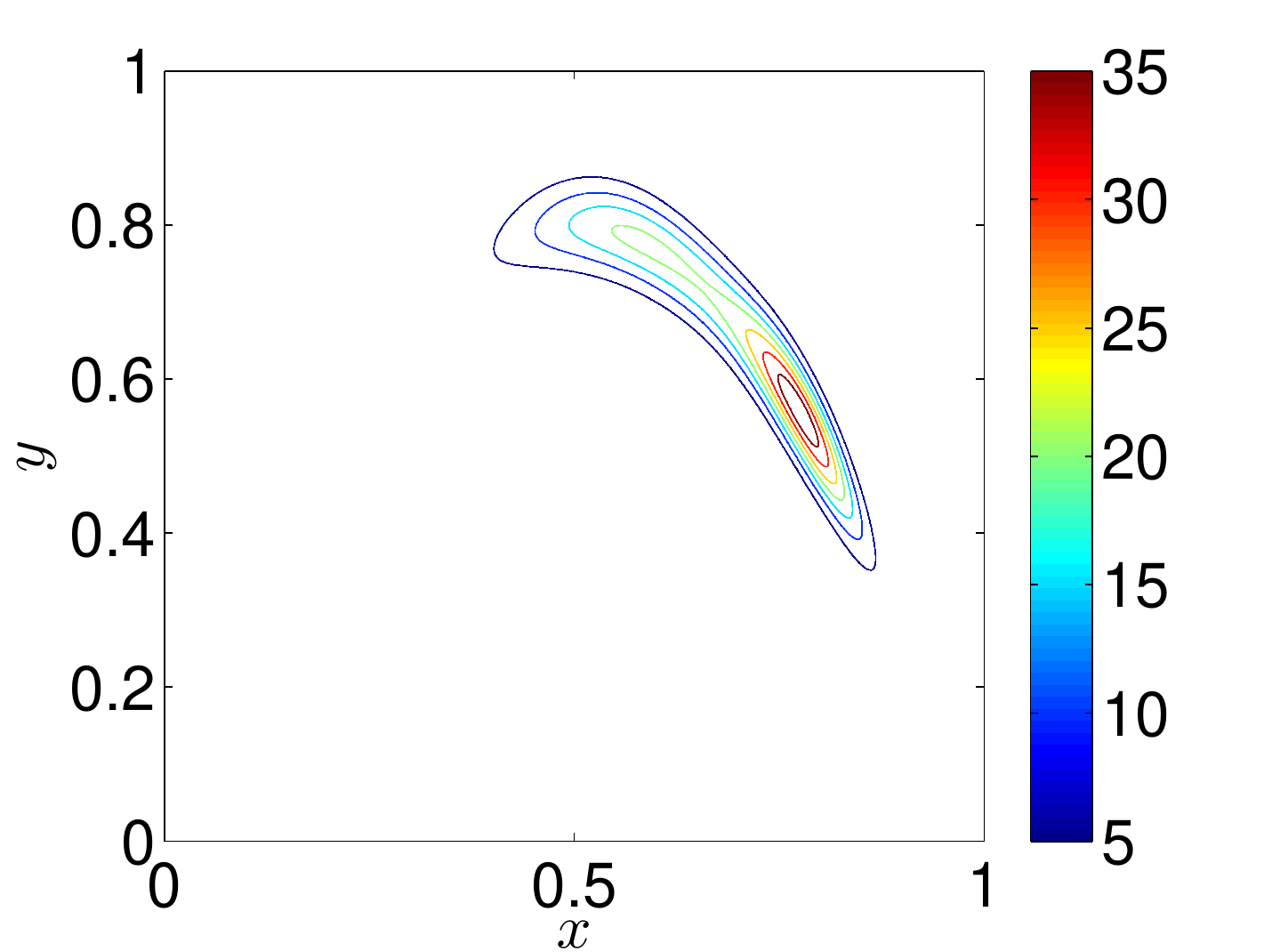}
		\caption{$\pi(x)P(x,y)$.}\label{fig:P-pi-J-B}
	\end{subfigure}

	\begin{subfigure}[b]{0.49\textwidth}
		\includegraphics[width=\textwidth]{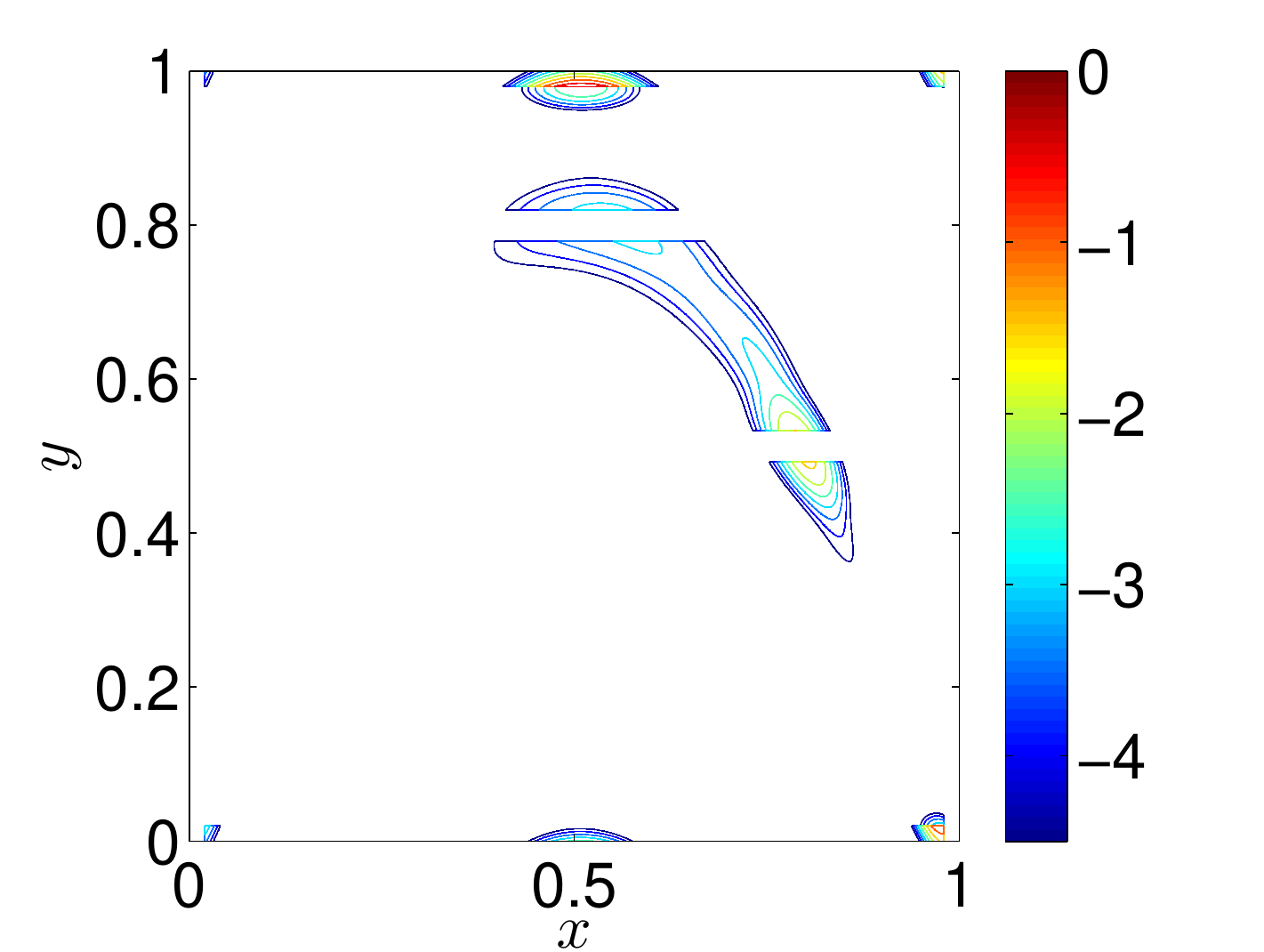}
		\caption{ $J(x,y)$}\label{fig:P-pi-J-C}
	\end{subfigure}
	~
	\begin{subfigure}[b]{0.49\textwidth}
		\includegraphics[width=\textwidth]{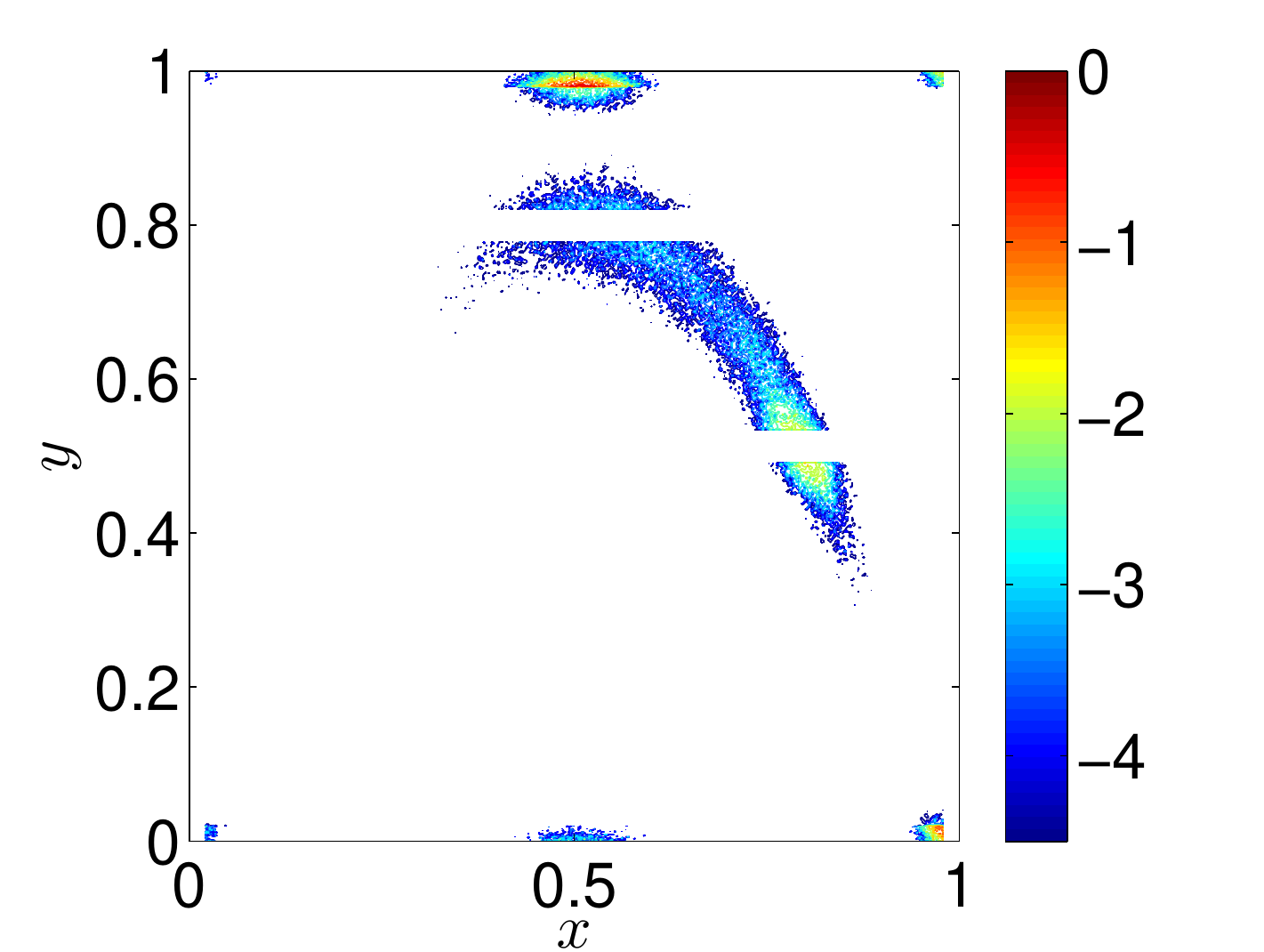}
		\caption{empirical $J(x,y)$.}	\label{fig:P-pi-J-D}
	\end{subfigure}

	\caption{ The plots of the transition kernel $P(x,y)$, the PDF flux $\pi(x)P(x,y)$ used in \cite{Billings::virus2002} and  the $A$-$B$ reactive current  $J(x,y)=\pi(x)P(x,y)q^-(x)q^+(y)$. The contour plots for  $J$ in subplots (C) and (D)  are
	actually  for the value   $\log(J(x,y)/M)$ where $M = \max_{x,y\in S} J(x,y)$ for
	    visualization.
	 The parameters are $\alpha = 3.2$, $\sigma = 0.04$, $\delta_a = \delta_b = 0.02$.
	 ($M=6.5186\times 10^{-4}$).}
	\label{fig:P-pi-J}
\end{figure}

{\it (3) $A$-$B$ reactive current.}
The transition kernel $P(x,y)$ is  shown in
Figure \ref{fig:P-pi-J-A}.
Figure \ref{fig:P-pi-J-B} plots $\pi(x)P(x,y)$,
which is the so-called ``PDF flux" in
 \cite{Billings::virus2002}.
 The $A$-$B$ reactive current in the TPT for our use,
   shown in Figure \ref{fig:P-pi-J-C},
 was calculated from \leqref{eqn:J}
 via solving \leqref{eqn:q-t} and \leqref{eqn:q+}
 by finite difference method.
 Figure \ref{fig:P-pi-J-D} is the empirical result from
 the direct simulation, which shows that
 our calculation is reliable.

\subsubsection{Stochastic instability comparison at $\alpha=3.08$}
We fix $\sigma = 0.04$ for the following discussion
about the transition mechanism at $\alpha=3.08$,
in which the period-2 orbit is $\po=(\pop_1,\pop_2)=(0.5696,0.7551)$.

\begin{table}[htbp]
	\begin{tabular}{|c|c|c|}
		\hline
		$k_{AB}$ & $\delta_a = 0.01$ & $\delta_a = 0.015$ \\
		\hline
		$\delta_b = 0.01$ & 4.6883$\times10^{-9}$ & 4.6883$\times 10^{-9}$ \\
		\hline
		$\delta_b = 0.015$ & 7.6215$\times 10^{-9}$ & 7.6215$\times 10^{-9}$ \\
		\hline
	\end{tabular}
	\caption{Transition rate $k_{AB}$ for different $\delta_a$ and $\delta_b$. Here, $\alpha = 3.08$, $\sigma = 0.04$.  }
	\label{tb::transition_rate}
\end{table}
 The first viewpoint of MPLP is to compare the total current out of $A$,
$r_{AB}^{-}(x)=\int_{\dom} J(x,y)\d y $ for $x\in A$.
The set $A$ of concern is the union $A_1\cup A_2$,
where $A_i= [\pop_i-\delta_a, \pop_i+\delta_a], ~i=1,2$.
The set $B=[0,\delta_b]\cup [1-\delta_b, 1]$.
Table \ref{tb::transition_rate}
shows that
  $\delta_a$, the width of the set $A$, has little influence
on the result of the transition rate $\kappa_{AB}$,
and  $\delta_b$ has a slightly more significant  influence
on $\kappa_{AB}$. This observation is expected since
the set $A$ is a small neighbourhood of the  linearly {stable}
periodic orbit of  the logistic map.
To test the impact on the MPLP point,
we plot  in Figure \ref{figure::p_AB} the
 total current
$r_{AB}^{-}(x)$ for  $x\in A_1$ ({\it left}) and $x\in A_2$ ({\it right})
for the  different  widths  specified in Table \ref{tb::transition_rate}.
As shown in this figure, the window $A_2$ where the periodic point
 $\pop_2$ lies carries $30\%\sim 50\%$ more reaction current
 than the window  $A_1$, for various values of $\delta_a$ and $\delta_b$.
 We also tested this result of the MPLP point
 by varying  $\sigma$ between $0.01$ and $0.04$, and reached
 the same conclusion that the second periodic point
 $\pop_2$ is the MPLP.

So, our technique based on the relative size
of the total current out of the set $A$ robustly
identifies   the point $\pop_2$ from the period-2 orbit $(\pop_1, \pop_2)$
as the MPLP. In the sense of the $A$-$B$  transition events,
we can claim that the point $\pop_2$ is less stable, or more active,
under the random perturbation.
Note that in terms of  the invariant measure,  $\pi(\pop_2)>\pi(\pop_1)$.
It is $\pop_1$ that has a smaller  equilibrium probability density.

 \begin{figure}[h!]
	\centering
 		\includegraphics[width = \textwidth]{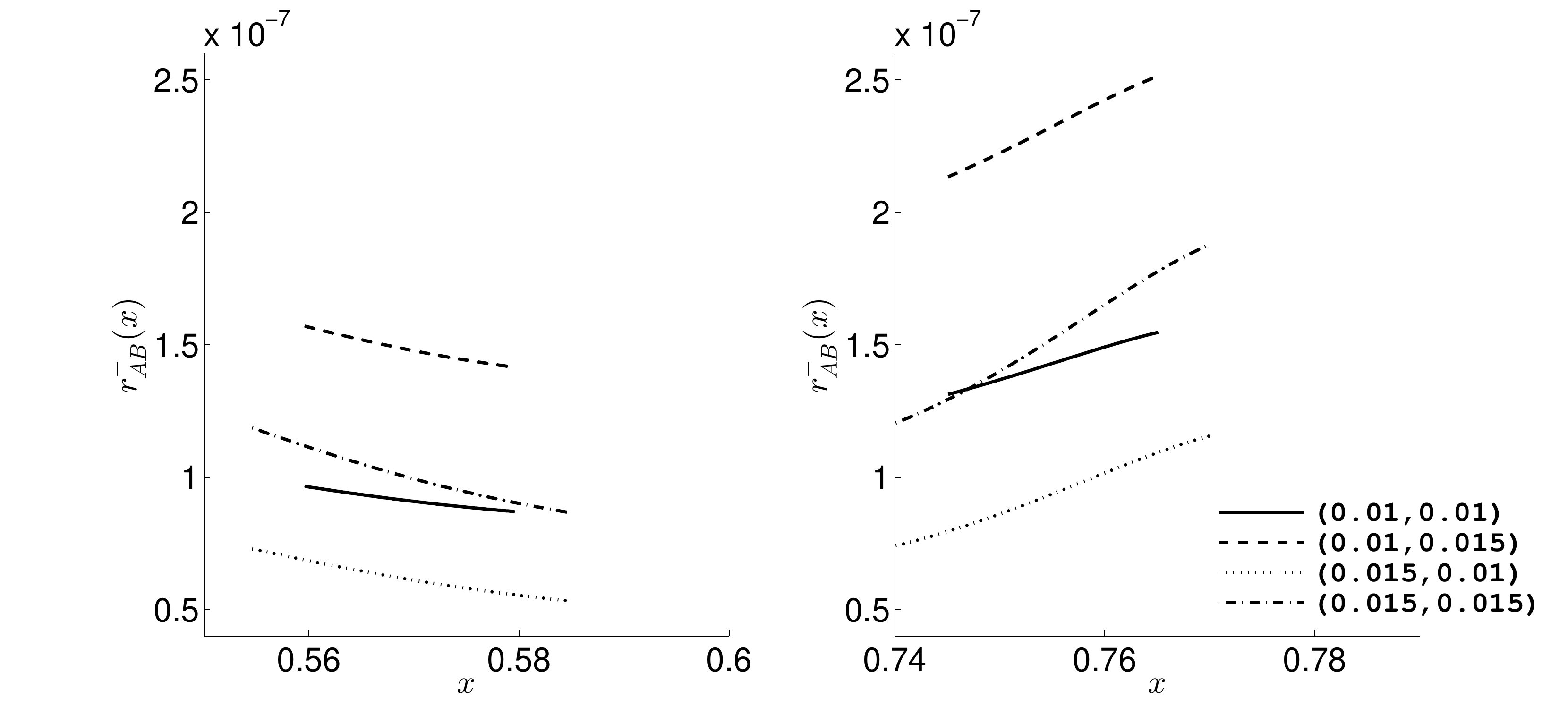}
	 	\caption{ $r_{AB}^{-}(x)$  for $x$ in the union of the sets  $A_1=[\pop_1-\delta_a, \pop_1+\delta_a]$(left) and $A_2= [\pop_2-\delta_a, \pop_2+\delta_a]$(right).
		 $\po=(\pop_1,\pop_2)=(0.5696,0.7551)$ is the period-2 orbit.
		  $\alpha=3.08$. $\sigma=0.04$.
		The solid line corresponds to $\delta_a = 0.01$, $\delta_b = 0.01$; the dashed line corresponds to $\delta_a = 0.01$, $\delta_b = 0.015$; the dotted line corresponds to $\delta_a = 0.015$, $\delta_b = 0.01$; the dash-dot line corresponds to $\delta_a = 0.015$, $\delta_b = 0.015$.
}	\label{figure::p_AB}
\end{figure}

In the following, we analyze the dynamical bottleneck and dominant transition pathways for this period-2 case.
 We will  restrict   to those dominant  transition paths with the minimal path lengths
 $N(z^*)$
to exclude the possible existence  of loops.
For simplicity, we will omit the $N(z^*)$ in the notation $W^{N(z^*),i}_{z^*}$ and write $W^i_{z^*}$.
  We choose the window width  $\delta_a = \delta_b = 0.01$.
 After building the effective reactive current  $J^+(x,y)$  ,
 we found that  the $A$-$B$ competency $z^{*}\approx 1.98\times 10^{-6}$
 by the binary search between $0$ and $M=\max_{\dom \times \dom} J^+(x,y)$.
 $N(z^*)$ is equal to $2$.
 Then, we look for the sequences of the sets $W^i_{z^*}$ for $i=2,1,0$,
 by using a number of pilot points to explore these sets.
 The numerical result, up to the accuracy $10^{-4}$, shows  the following:
 \[
 \begin{split}
  W^{2}_{z^{*}} &= [0.9900, 0.9928]\subset B,
 \\
W^{1}_{z^{*}} &= \{0.5331\}
\subset \dom \setminus (A\cup B),
\\
     W^{0}_{z^{*}} &= \{0.7651\}\subset A.
      \end{split}
     \]
  Then the
 $A$-$B$ dynamical bottleneck $\Btn (A,B)$
 is $(0.7651,0.5331)$.
 Let   $0.5331$ be the new set $A'$ and search for the  $A'$-$B$ dynamical bottleneck. Then we obtain the second dynamical bottleneck $(0.5331, 0.9900)$.
 Finally, we get the dominant transition path $$\pt \approx (\underline{0.7651}, \underline{0.5331}, 0.9900), ~~~\mbox{at}~~~ \sigma=0.04,$$
 where the underlined values  correspond to the location of the dynamical  bottlenecks.
This result of the dominant transition path is unchanged when we changed the grid size
between $1.7\times 10^{-4}$  and  $3.4\times 10^{-4}$ in discretizing the space $\dom=[0,1]$.
We also varied the width $\delta_a$ between $0.01$ and $0.02$
and obtained the same  result for the dominant transition path $\vec{\varphi}$.
The first point of the dominant transition path $\vec{\varphi}$,
i.e., the point in $W^0_{z^*}$, lies in the window $A_2$
for the second periodic point $\xi_2$. Thus the $A$-$B$ competency
is actually realized by the $A_2$-$B$ competency.
  So, we conclude that $\xi_2$ is   also the MCPP.
    The $A$-$B$ dominant transition path
 starts from a boundary point in $A_2$, followed by a jump to some point  on the left but far away from
$\pop_1$  to escape the periodic orbit, and eventually jumps into the set $B$.

\subsection{Bifurcation diagram for the period-2 case}
It is interesting to see how the above transition mechanisms
(MPLP, MCPP, dominant transition paths, etc.)
change when the noise amplitude $\sigma$ or the parameter $\alpha$ changes.
The following numerical results  show bifurcations for varying parameters,
and we will see that the two criteria
do not always give the same   conclusion.

\subsubsection{change $\sigma$}
We still fix $\alpha=3.08$ but now change the value of the noise amplitude $\sigma$
between $0.01$ and $0.04$.
Remind that the period-2 orbit is $\po=(\pop_1,\pop_2)=(0.5696,0.7551)$.

Figure \ref{fig:MPLP} plots the probability density at $\pop_1$ and $\pop_2$
of the last hitting distribution of the transitions from $A$ to $B$.
 It shows that  $\pop_2$ always wins $\pop_1$ as the MPLP
 for $\sigma\in (0.01,0.04)$.
%We did not observe the exchange of this MPLP.
The dependence of the transition rate $\kappa_{AB}$ on the noise amplitude $\sigma$,
 in Figure \ref{fig:Arrh},
shows an Arrhenius-like relation.
\begin{center}
\begin{figure}[htbp]
\begin{subfigure}[b]{0.475\textwidth}
		\includegraphics[width=\textwidth]{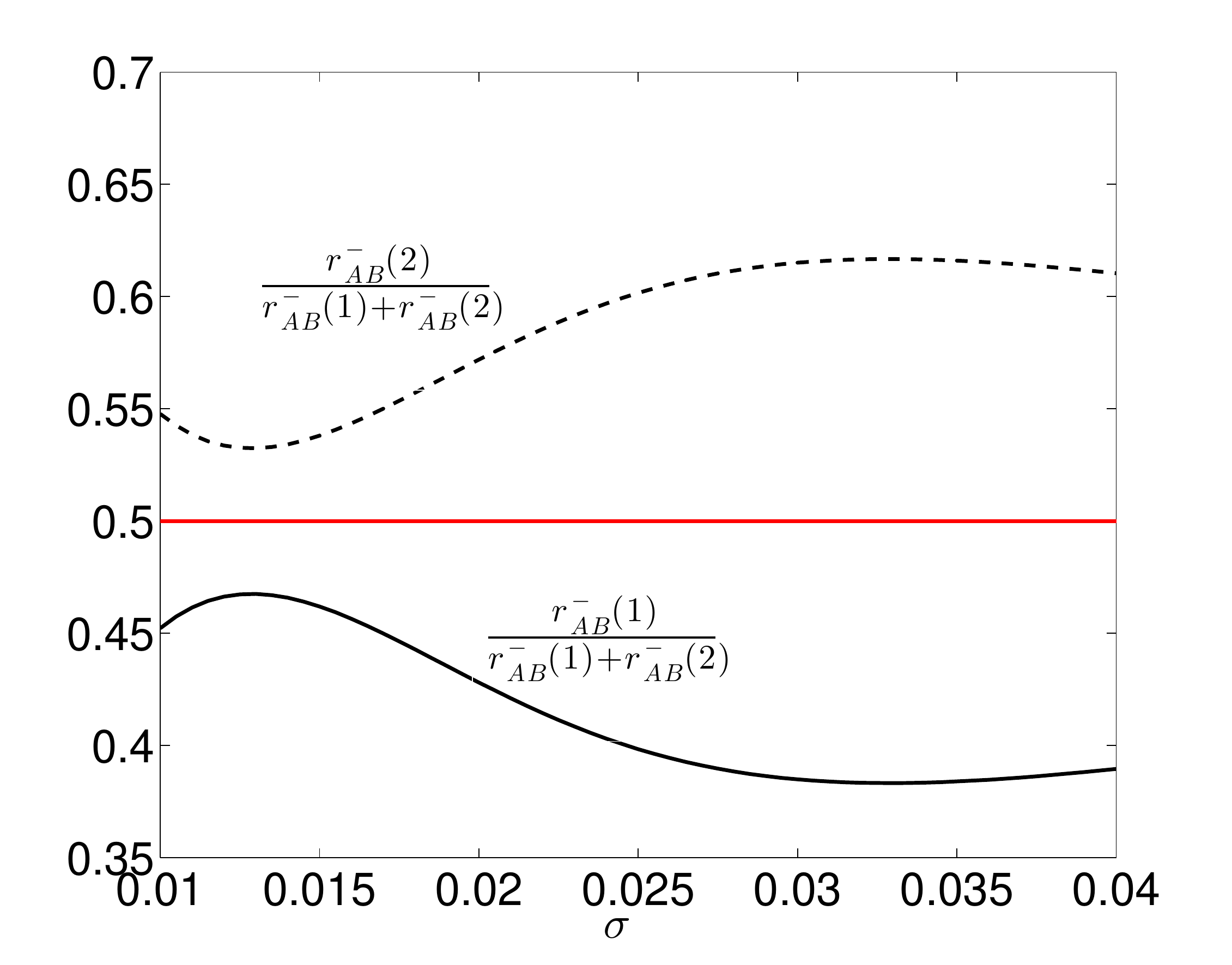}
		\caption{$\frac{r^{-}_{AB}(i)}{r^{-}_{AB}(1)+r^{-}_{AB}(2)}$ versus $\sigma$}
		\label{fig:MPLP}
	\end{subfigure}
	\begin{subfigure}[b]{0.495\textwidth}
		\includegraphics[width= \textwidth]{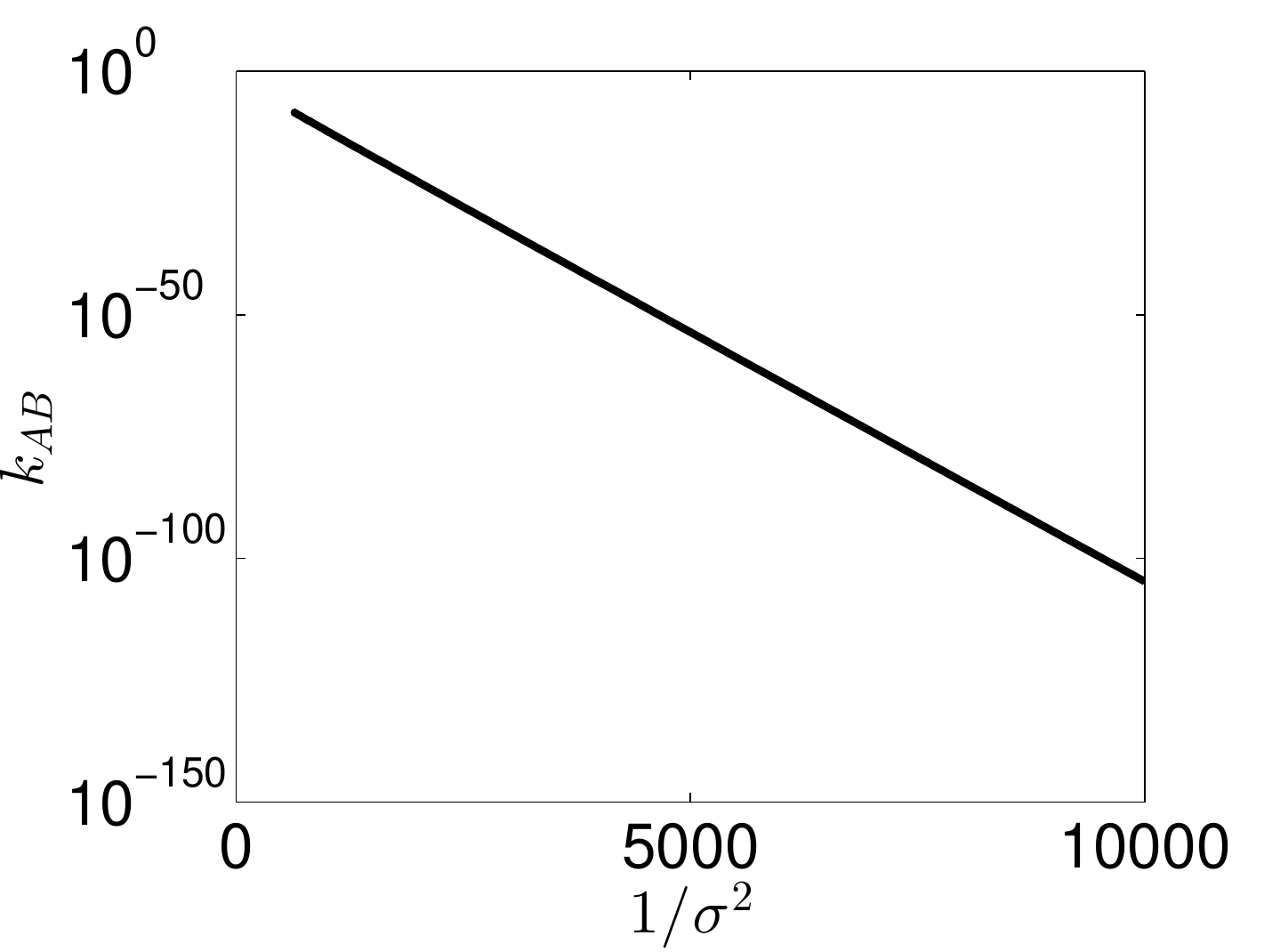}
		\caption{$\kappa_{AB}$ versus $1/\sigma^{2}$. }
		\label{fig:Arrh}
 	\end{subfigure}
\end{figure}
\end{center}

For the results about the dominant  transition paths,
the first observation is that  the minimal length of the dominant transition paths,
$N(z^*)$, grows as $\sigma$ decreases.
 For example,
 at  $\sigma=0.02$, the dominant transition path is
 $\pt= (\underline{0.7651},
\underline{0.5269}, 0.9743, 0.0100).$
At $\sigma = 0.014$,  the dominant transition path has  the minimal length $5$:
$$\pt = (\underline{0.5596}, \underline{0.7761}, 0.5181, 0.9740, 0.0100).$$
 At $\sigma = 0.013$,  the dominant transition path has  the minimal length $6$:
 $$\pt = (0.7643, \underline{0.5508},\underline{0.7818}, 0.5118, 0.9740, 0.0100).$$

\begin{center}
\begin{figure}[htbp]
		\includegraphics[width=.6\textwidth, height=0.3\textwidth]{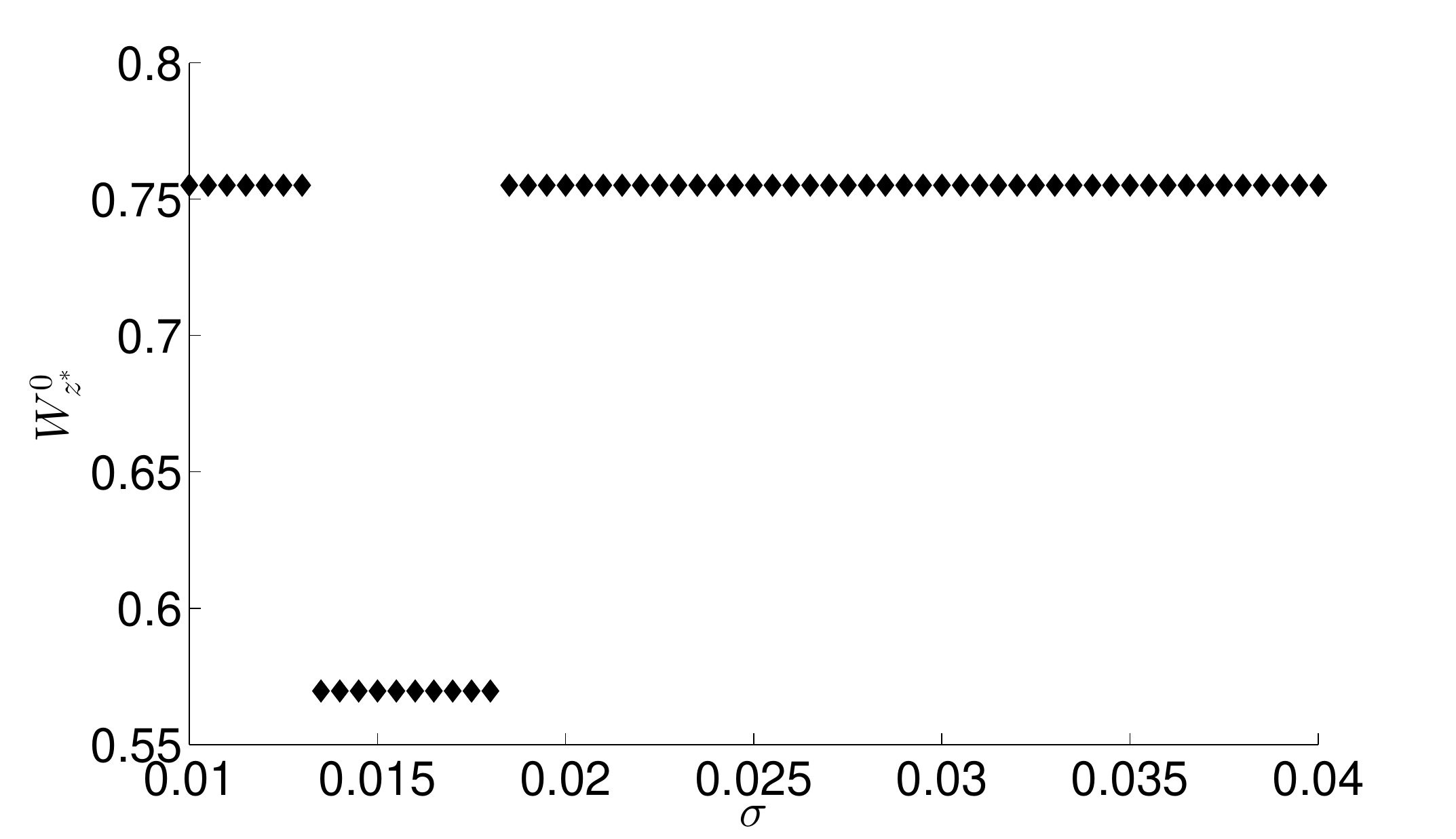}
		\caption{ The maximum competency periodic point. }
		\label{fig:pps}
\end{figure}
\end{center}

When $\sigma$ varies, Figure \ref{fig:pps} plots the MCPP
among the choices of the periodic points $\pop_1$ and $\pop_2$.
This figure shows  two critical values of $\sigma$:
$\sigma_1\approx 0.0134$ and $\sigma_2\approx 0.0185$,
where $\pop_1$ and $\pop_2$ exchange their roles as
MCPP.

\begin{center}
 \begin{footnotesize}

	 \newcommand{\twa}{0.12\textwidth}
 \newcommand{\twb}{0.21\textwidth}

\begin{longtable}{|p{\twa}|p{\twb}|p{\twb}|p{\twb}|}
 	\hline
	$\sigma=$ & $0.0130<\sigma_1$ &  $0.0134 \approx \sigma_1$&  $0.0136>\sigma_1$  \\

	\hline
	$\pt_1=(\pop_1,\cdots)$ & \underline{0.5596}, \underline{0.7751}, 0.5196, 0.9740, 0.0100 & \underline{0.5596}, \underline{0.7758}, 0.5186, 0.9740, 0.0100  & \underline{0.5596}, \underline{0.7758}, 0.5186, 0.9740, 0.0100  \\
	\hline
	$J^+(\pt_1)$ & {\bf 0.0150}, 0.0165, 0.0242, 0.3753 &  {\bf 0.0053}, 0.0063, 0.0093, 0.1255 & {\bf 0.0243}, 0.0289, 0.0424, 0.5557   \\
	
		\hline

	\hline
	$\pt_2=
	(\pop_2,\cdots)$ & 0.7644, \underline{0.5509}, \underline{0.7818}, 0.5119, 0.9740, 0.0100& 0.7641, \underline{0.5513}, \underline{0.7818}, 0.5119, 0.9740 , 0.0100 & 0.7641, \underline{0.5516}, \underline{0.7818}, 0.5119, 0.9740, 0.0100 \\
	\hline
	$J^+(\pt_2)$ & 0.0230, {\bf 0.0152}, 0.0260, 0.0363, 0.3753 & 0.0079, {\bf 0.0053}, 0.0091, 0.0126, 0.1255 & 0.0357, {\bf 0.0240}, 0.0410, 0.0567, 0.5557  \\
		\hline
	\hline
	$z^*=$ & 0.0152  & 0.0053 & 0.0243
	\\
	\hline
	$\pt=$ & $\pt_2$ &$\pt_1$, $\pt_2$   & $\pt_1$  \\
	\hline
	\caption{The $A$-$B$ dominant transition path $\pt$
	and the  $A$-$B$ competency  $z^*$ for three values of $\sigma$:
	$\sigma<\sigma_1$, $\sigma=\sigma_1$ and $\sigma>\sigma_1$.
	  $\pt_1$ and $\pt_2$ are the $A_1$-$B$  and $A_2$-$B$ dominant transition paths
	  respectively.  The various dynamical bottlenecks are   underlined.  Their capacities are marked in bold font  at the $J^+$ row (by multiplying the unit $10^{-58}$,   $10^{-54}$  and $10^{-53}$, respectively for each column).
	 The   $A$-$B$ competency,  $z^*$, is  determined by the maximum of the $A_1$-$B$ and $A_2$-$B$ capacities. Note that the periodic points are located at
	 $\po=(\pop_1,\pop_2)=(0.5696,0.7551)$ and $\delta_a=\delta_b=0.01$.}
	%	weight scale & $10^{-58}$ & $10^{-54}$ & $10^{-53}$   \\

	\label{table::bifurcation}
\end{longtable}
 \end{footnotesize}

\end{center}

We demonstrate  a more detailed
analysis at the bifurcation point $\sigma_1$
 in Table \ref{table::bifurcation}
as well as in Figure \ref{figure::path}.
Table \ref{table::bifurcation} compares
the $A_i$-$B$  dominant  transition paths $\pt_i$, $i=1,2$.
That is   $z^*(A_i, B)=\Cp(\pt_i)$ for $i=1,2$.
The $A$-$B$ dominant transition path
$\pt$ is the path among $\pt_1$ and $\pt_2$ with the larger competency.
Remind that the competency of a given path $(\varphi_0,\cdots,\varphi_N)$ is calculated as
the minimum of the effective currents $J^+(\varphi_n,\varphi_{n+1})$
at each edge $(\varphi_n,\varphi_{n+1})$,
which is denoted in bold font in Table~\ref{table::bifurcation}.

\begin{figure}
		\begin{subfigure}[b]{0.8\textwidth}
		\includegraphics[width=\textwidth]{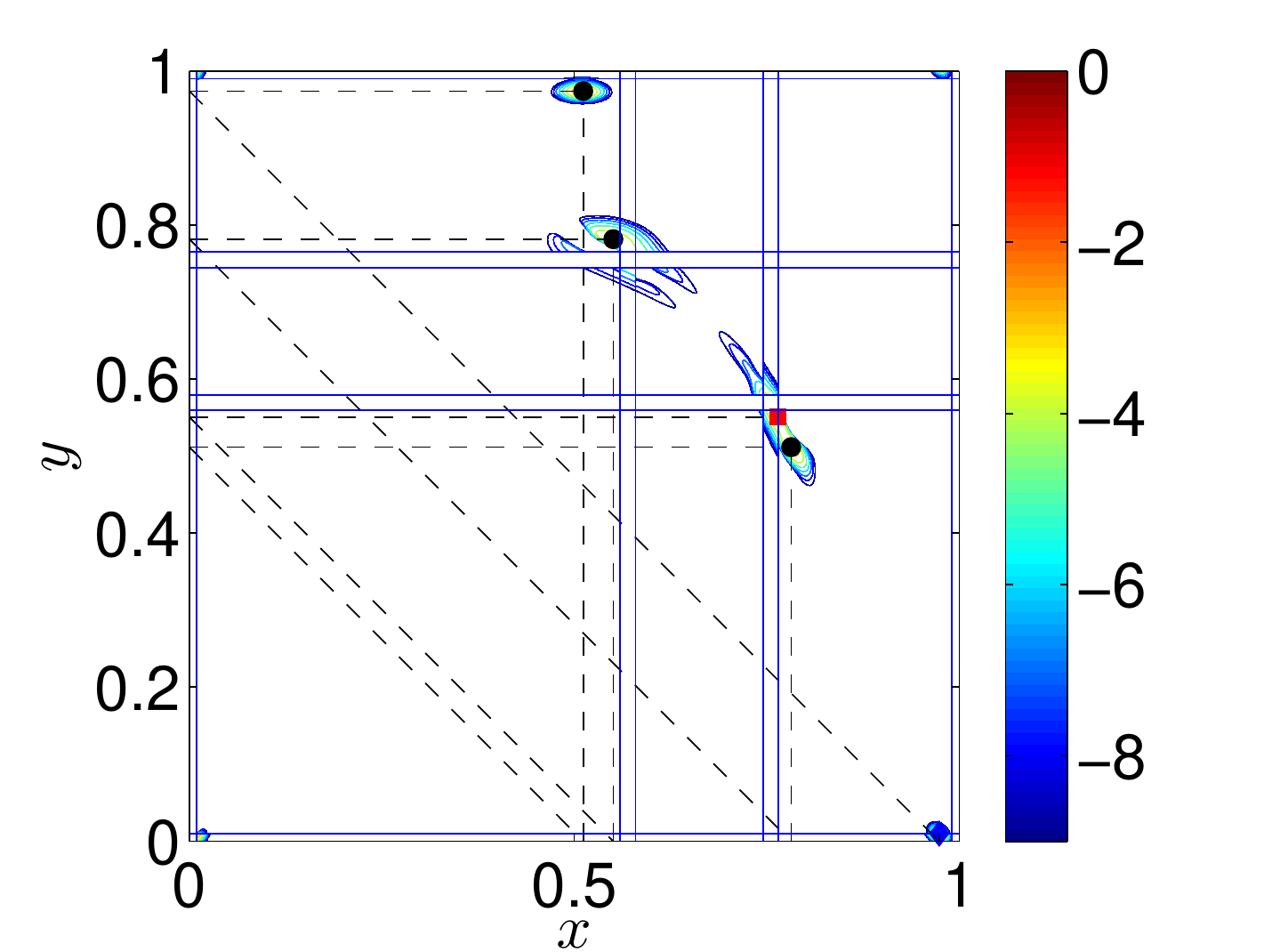}
	%	\caption{$\sigma = 0.013$. \\ $\pt = (0.7643, 0.5508, 0.7818, 0.5118, 0.9740, 0.0100)$}
		\label{figure::path_0_013}
	\end{subfigure}
	\begin{subfigure}[b]{0.8\textwidth}
		\includegraphics[width=\textwidth]{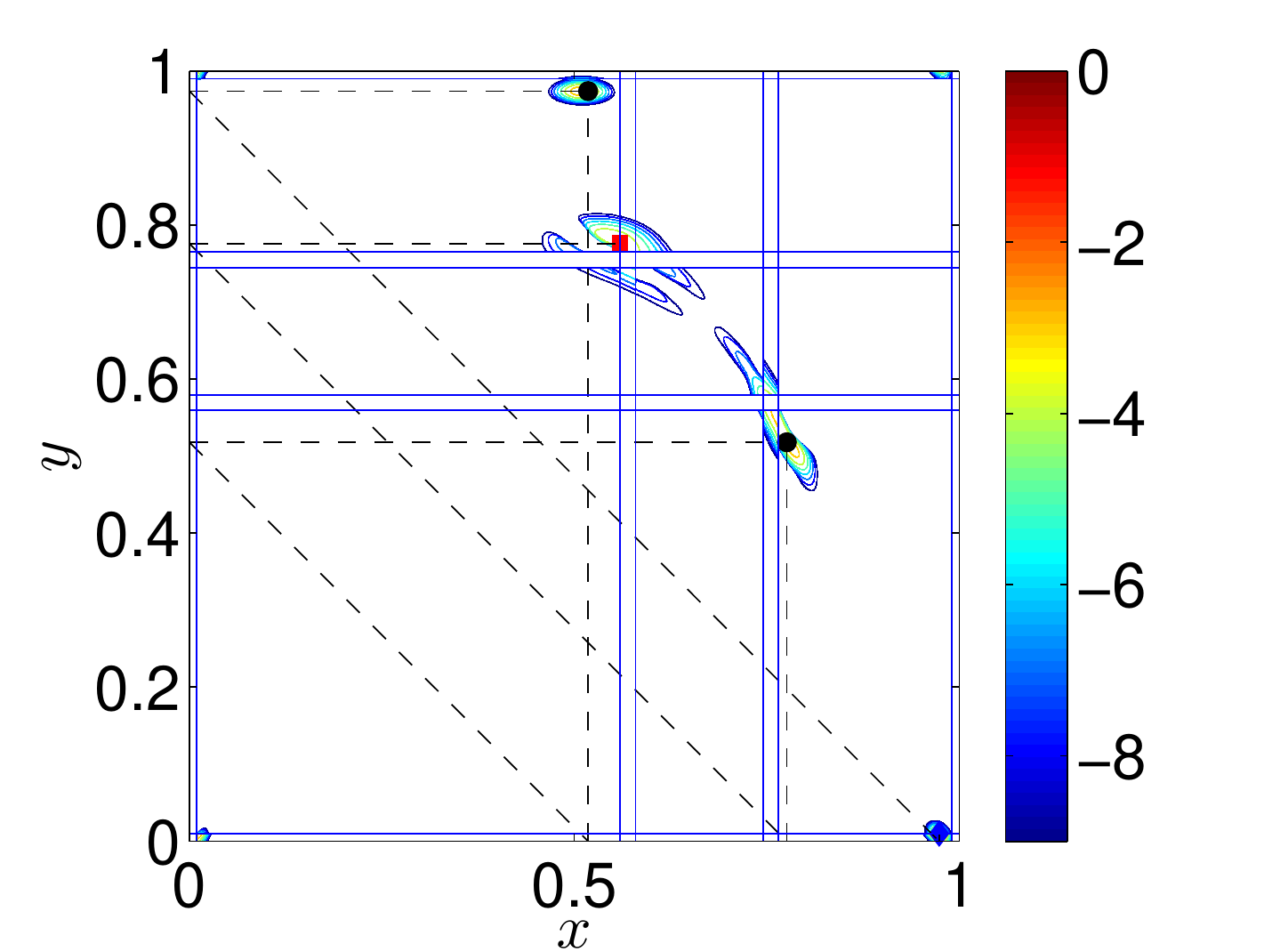}
	%	\caption{$\sigma = 0.014$.\\ $\pt = (0.5596, 0.7761, 0.5181, 0.9740, 0.0100)$}
		\label{figure::path_0_014}	
	\end{subfigure}
	\caption{This figure  visualizes (in form of cobweb plot) the dominant transition path in the contour plot of $\log(J^{+}/M)$ where $M = \max_{x,y\in \dom} J^{+}(x,y)$ ($\alpha = 3.08$, $\delta_a=\delta_b=0.01$).
	Top: $\sigma=0.013$; Bottom: $\sigma=0.014$.
	 The six vertical and six horizontal straight lines (solid, blue) are   the boundaries of   $A$ and $B$.
	 The red and blue dots represent the first  and the last edge of the path.
	 The black dots represent all the other edges.}
	\label{figure::path}
\end{figure}

\begin{figure}
	\includegraphics[width=0.6\textwidth]{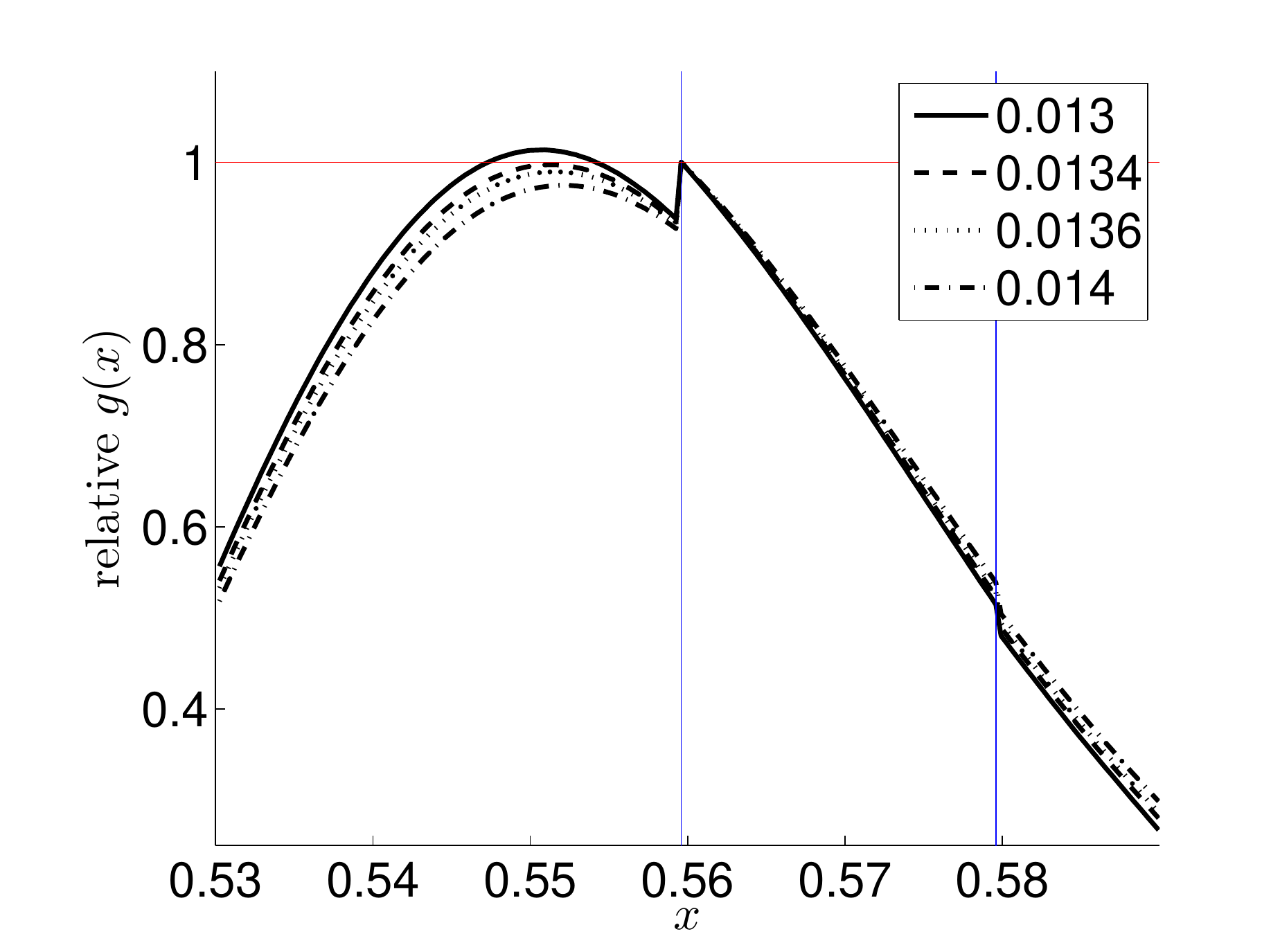}
	\caption{ The plot of      $g(x)/g(\pop_1-\delta_a)$  near $A_1$ for
	the four values of $\sigma$ from $0.013$ and $0.014$. Note that
	$g$ is not continuous at the boundary locations of   $A_1$:
	$\pop_1-\delta_a=0.5596$  and $\pop_1+\delta_a=0.5796$, shown as the two  vertical lines in this figure. }
	\label{fig:gb}
\end{figure}

To understand the bifurcation of the   MCPP,
we need analyze the competition of the
two capacities $z^*(A_1,B)$ and $z^*(A_2,B)$,
which are further determined by the $A_i$-$B$ dynamical
bottlenecks $\Btn(A_i,B)$ on the $A_i$-$B$ dominant transition paths $\pt_i$ for  $i=1,2$.
The $A_1$-$B$ dynamical bottleneck is the first step of jump on $\pt_1$,    from
the left boundary point of $A_1$ to a point
(located at $0.77\sim0.78$) near
the right interval $A_2$.
The $A_2$-$B$ dynamical bottleneck is
the   second step on $\pt_2$, corresponding to
the jump from a point  slightly on the left  side of  the interval $A_1$,
to a point quite close to one point of the  $A_1$-$B$ dynamical bottleneck.
Hence for $\sigma$ around the value $\sigma_1$,
both of the $A_i$-$B$, $i=1,2$ dynamical bottlenecks are
the jumps from a region near the left boundary  of $A_1$
(including $A_1$'s left boundary), denoted as $I_1$
to a region near the right boundary of $A_2$, denoted as $I_2$.
So, by setting $I_2=[0.7, 0.82]$ and $I_1=[0.53,0.59]$,
we investigate  the maximum possible reactive current  for any given $x\in I_1$:
 $g(x):=\max_{y\in I_2} J^{+}(x,y)$ for $x\in I_1$.
 The maximizer of this function, whether it is equal to the left boundary of $A_1$ or not,
 will determine which one of $\pt_1$ and  $\pt_2$ is the $A$-$B$ dominant transition path.
  By plotting the graph of the function $g$
  for several $\sigma$ values around the critical value $\sigma_1$
  in Figure \ref{fig:gb} and rescaling $g$ by its value at the left boundary of $A_1$,
  we indeed find that it is the competition of two local maximizers
  of $g$ that leads to the bifurcation of the dominant transition path
  from $\pt_2$ to $\pt_1$ as $\sigma$  increasingly   passes $\sigma_1$ .

  The bifurcation at the second critical value $\sigma_2$ of the noise amplitude  $\sigma$
is also due to the change of the effective current  $J^+(x,y)$,
which yields the changes of the MPCC and the dominant transition paths,
   via  the competition of the local maximizers in the interiors
and the  values  at the  boundary points of  $A_1$ and $A_2$
for the function $J^+(x,y)$.

\subsubsection{Change $\alpha$}

\begin{figure}[htbp]
\includegraphics[width=0.48\textwidth]{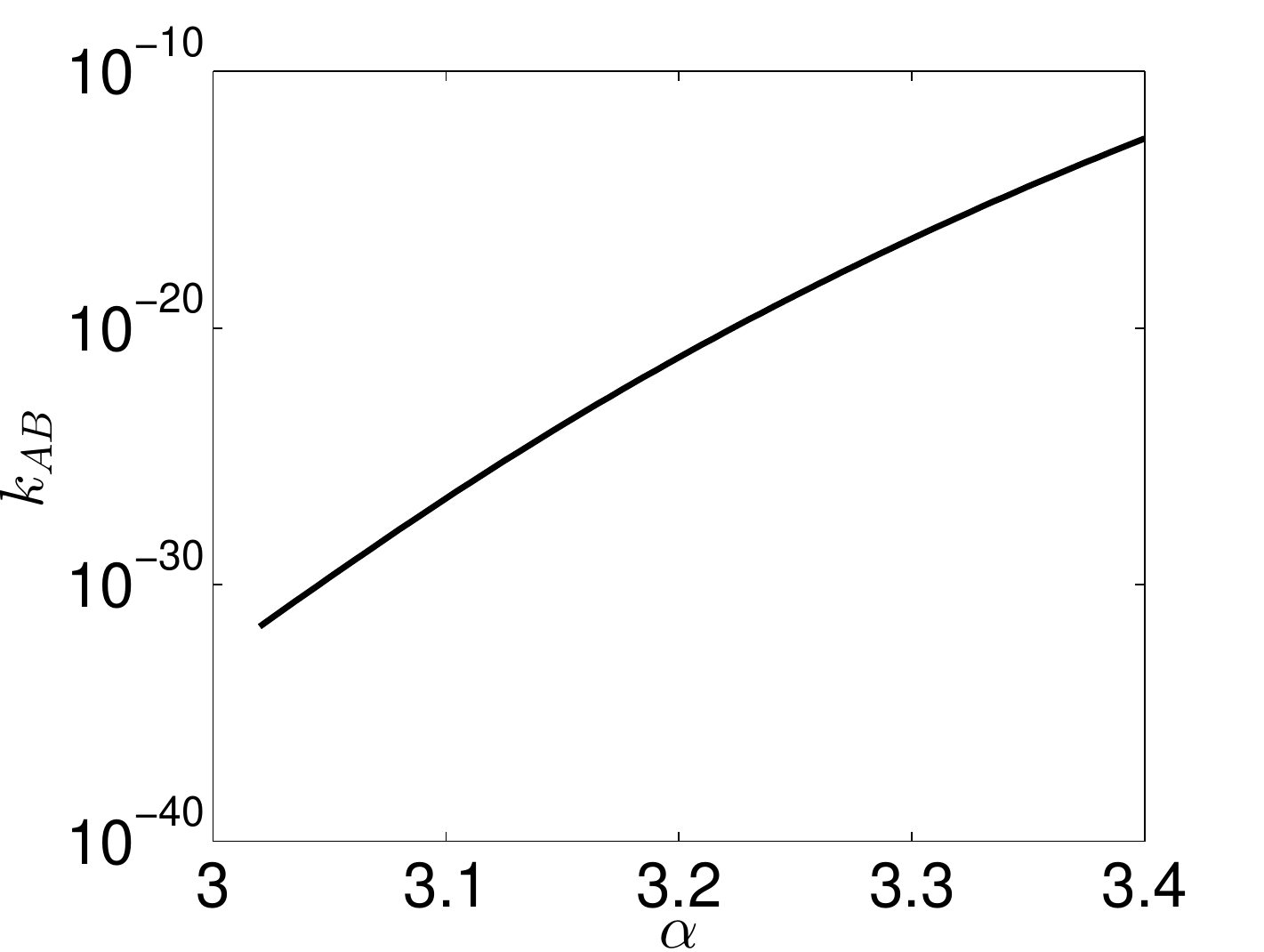}
\caption{$\kappa_{AB}$ versus $\alpha$.}
\label{fig:kalpha}
\end{figure}

\begin{figure}[htbp]
	\centering
	\begin{subfigure}[b]{0.48\textwidth}
		\includegraphics[width=\textwidth]{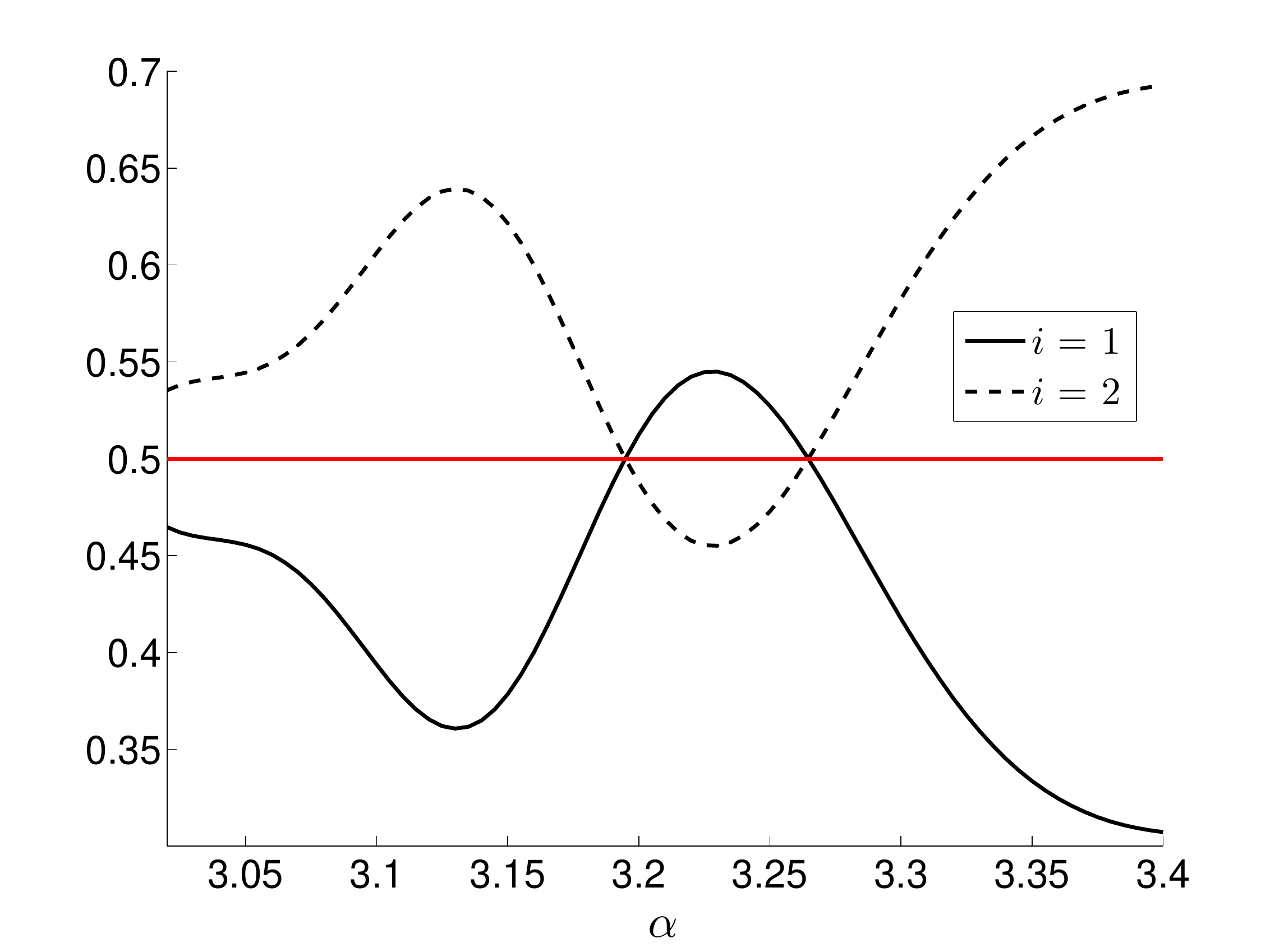}
		\caption{${r_{AB}^-(i)} /{( r_{AB}^-(1)+r_{AB}^-(2))}$ versus $\alpha$.}
		\label{fig:ralpha}
	\end{subfigure}
	~
	\begin{subfigure}[b]{0.48\textwidth}
		\includegraphics[width=\textwidth]{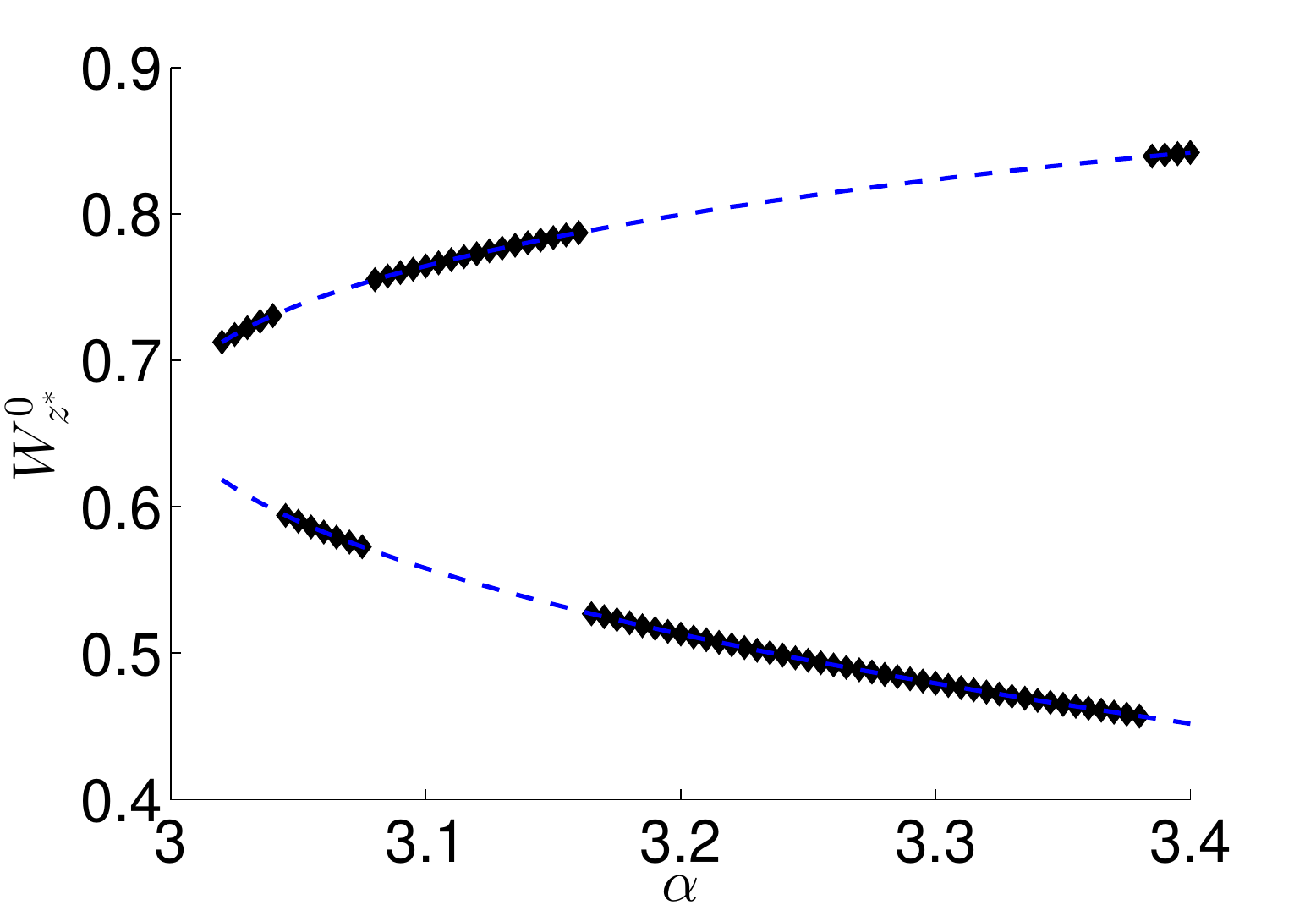}
		\caption{The MCPP.}
		\label{fig:palpha}
	\end{subfigure}
	\caption{ The MPLP and the MCPP
	for $3.02\leq \alpha \leq 3.40$. $\sigma = 0.02$, $\delta_a = \delta_b = 0.01$.
	The horizontal straight line in (A) indicates the threshold $0.5$.
		The dashed curves in  (B) represent  the locations of the two periodic points for each $\alpha$.}
	\label{fig::varying_alpha}
\end{figure}

When $\alpha \in [3.02, 3.40]$, the only stable invariant set of the \Lm\
is the period-2 orbit. We are now interested in how the
value of $\alpha$ influences the
transition rate and the roles of the individual periodic points.
Fix $\sigma=0.02$ and $\delta_a=\delta_b=0.01$.
Figure \ref{fig:kalpha} shows that the transition rate $\kappa_{AB}$
 increases in $\alpha$ and this dependency is
 nearly exponential.
 To identify the MPLP between the two periodic points $\pop_1$
 and $\pop_2$ ($\pop_1$ is defined to be the smaller one),
 the   probability mass $r_{AB}^-(i) / (r_{AB}^-(1)+r_{AB}^-(2))$
 is plotted in Figure \ref{fig:ralpha}.
 As shown in this figure,  $\pop_1$ is the  MPLP
 only when  $\alpha $ is  approximately  between $3.20$ and $3.26$.
  Figure \ref{fig:palpha} shows the
  MCPP in dark diamond-shaped  dots
for each $\alpha$.  For the range of $\alpha$ we investigated here,
there are four critical values of $\alpha$ where the MCPP
switches between the  two periodic points $\pop_1$ and $\pop_2$.
As explained in Section \ref{ssec:two},  the MPLP and MCPP
can be different so the bifurcation points of $\sigma$
in Figure  \ref{fig:ralpha} and  \ref{fig:palpha} are different  .

%\subsection{Results for the period-3 case  }

\section{Discussion}
\label{sec:sum}
In conclusion, we have described the method based on the \tpt ,
illustrated on the example of the randomly perturbed \Lm, to study the
stochastic  instability of the linearly stable periodic orbit
 in the context of noise-induced transitions.
The introduced concepts of most-probable-last-passage point and the
maximum competency point   are the novel descriptions of the
stochastic instability for the  linearly stable periodic orbit.
We demonstrated the capability
of these two proposed perspectives
to quantify the stochastic instabilities of the
individual periodic point in one periodic orbit.
It should be noted that
although   only    the case of period-2  in discrete map
was analysed here,
our method can also be applied to other types of the
set $A$ with more complex structures.
In fact,
our approach  based on the \tpt\  is  generic to any ergodic stochastic dynamical systems,
 such as the multiplicative random perturbations,
 and to the arbitrary nonintersecting closed subsets  $A$ and $B$,
such as  the stable limit cycles in   continuous-time dynamical systems.

%
%We   numerically investigate the situations where the noise amplitude  $\sigma$
%or the parameter $\alpha$ in \Lm\ changes.  Even in the regime of the period-2 case,
%the bifurcation was observed for varying $\sigma$ and $\alpha$.
%Figure ~\ref{fig:pps} shows that for $\sigma$ as small as $0.01$,
%   the most probable
%$A$-transition point  is the right periodic point  $\pop_2$.
%If  $\sigma$ tends   to zero.
%it is an open question
%whether this  MCPP  be always    $\pop_2$ or it would  oscillate
%between $\pop_1$ and $\pop_2$.

\bibliographystyle{ieeetr}

%amsalpha
\bibliography{reference,gad}

\begin{thebibliography}{10}

\bibitem{Kramers}
H.~A. Karmers, ``Brownian motion in a field of force and the diffusion model of
  chemical reactions,'' {\em Physica}, vol.~7, pp.~284--304, 1940.

\bibitem{Kampen}
N.~G. van Kampen, {\em Stochastic Processes in Physics and Chemistry}, vol.~1.
\newblock North Holland, 2~ed., 1992.

\bibitem{Erying}
H.~Eyring, ``The activated complex and the absolute rate of chemical
  reactions,'' {\em Chem. Rev.}, vol.~17, pp.~65--77, 1935.

\bibitem{FW1998}
M.~I. Freidlin and A.~D. Wentzell, {\em Random Perturbations of Dynamical
  Systems}.
\newblock Grundlehren der mathematischen Wissenschaften, New York:
  Springer-Verlag, 2~ed., 1998.

\bibitem{MatSchuss1983}
B.~J. Matkowsky, Z.~Schuss, and C.~Tier, ``Diffusion across characteristic
  boundaries with critical points,'' {\em SIAM J. Appl. Math.}, vol.~43, no.~4,
  p.~673, 1983.

\bibitem{Naeh1990}
T.~Naeh, M.~M.~K. osek, B.~J. Matkowsky, and Z.~Schuss, ``A direct approach to
  the exit problem,'' {\em SIAM J. Appl. Math.}, vol.~50, no.~2, pp.~595--627,
  1990.

\bibitem{Maier1992PRL}
R.~S. Maier and D.~L. Stein, ``Transition-rate theory for nongradient drift
  fields,'' {\em Phys. Rev. E}, vol.~69, no.~26, p.~3691, 1992.

\bibitem{Maier1993PRE}
R.~S. Maier and D.~L. Stein, ``Escape problem for irreversible systems,'' {\em
  Phys. Rev. E}, vol.~48, no.~2, pp.~931--938, 1993.

\bibitem{weinan-MAM2004}
W.~E, W.~Ren, and E.~Vanden-Eijnden, ``Minimum action method for the study of
  rare events,'' {\em Comm. Pure Appl. Math.}, vol.~57, pp.~637--656, 2004.

\bibitem{aMAM2008}
X.~Zhou, W.~Ren, and W.~E, ``Adaptive minimum action method for the study of
  rare events,'' {\em J. Chem. Phys.}, vol.~128, no.~10, p.~104111, 2008.

\bibitem{Heymann2006}
M.~Heymann and E.~Vanden-Eijnden, ``The geometric minimum action method: a
  least action principle on the space of curves,'' {\em Comm. Pure Appl.
  Math.}, vol.~61, pp.~1052--1117, 2008.

\bibitem{Lorenz2008}
X.~Zhou and W.~E, ``Study of noise-induced transitions in the {L}orenz system
  using the minimum action method,'' {\em Comm. Math. Sci.}, vol.~7,
  pp.~341--355, 2009.

\bibitem{KS-WZE2009}
X.~Wan, X.~Zhou, and W.~E, ``Study of noise-induced transition and the
  exploration of the configuration space for the {Kuromoto-Sivachinsky}
  equation using the minimum action method,'' {\em nonlinearity}, vol.~23,
  no.~3, 2010.

\bibitem{DykmanPRL1992}
M.~Dykman, P.~McClintock, V.~Smelyanski, N.~Stein, and N.~Stocks, ``Optimal
  paths and the prehistory problem for large fluctuations in noise-driven
  system,'' {\em Phys. Rev. Lett.}, vol.~68, no.~18, p.~2718, 1992.

\bibitem{KautzPRA1987}
R.~L. Kautz, ``Activation energy for thermally induced escape from a basin of
  attraction,'' {\em Phys. Rev. A}, vol.~125, pp.~315--319, 1987.

\bibitem{KautzPRA1988}
R.~L. Kautz, ``Thermally induced escape: the principle of minimum available
  noise energy,'' {\em Phys. Rev. A}, vol.~38, no.~4, pp.~2066--2080, 1988.

\bibitem{Graham1991}
R.~Graham, A.~Hamm, and T.~T\'el, ``Nonequilibrium potentials for dynamical
  systems with fractal attractors or repellers,'' {\em Phys. Rev. Lett.},
  vol.~66, no.~24, pp.~3089--3092, 1991.

\bibitem{Kraut2003}
S.~Kraut and U.~Feudel, ``Enhancement of noise-induced escape through the
  existence of a chaotic saddle,'' {\em Phys. Rev. E}, vol.~67, no.~1,
  p.~015204, 2003.

\bibitem{LuchJETP1999}
D.~G. Luchinsky and I.~A. Khonanov, ``Fluctuation-induced escape from the basin
  of attraction of a quasiattractor,'' {\em JETP Letters}, vol.~69, no.~11,
  pp.~825--830, 1999.

\bibitem{LuchPRL2003}
A.~N. Silchenko, S.~Beri, D.~G. Luchinsky, and P.~V.~E. McClintock,
  ``Fluctuational transitions through a fractal basin boundary,'' {\em Phys.
  Rev. Lett.}, vol.~91, no.~17, p.~174104, 2003.

\bibitem{LuchPRE2005}
A.~N. Silchenko, S.~Beri, D.~G. Luchinsky, and P.~V.~E. McClintock,
  ``Fluctuational transitions across different kinds of fractal basin
  boundaries,'' {\em Phys. Rev. E}, vol.~71, no.~4, p.~046203, 2005.

\bibitem{Billings::virus2002}
L.~Billings, E.~M. Bollt, and I.~B. Schwartz, ``Phase-space transport of
  stochastic chaos in population dynamics of virus spread,'' {\em Phys. Rev.
  Lett.}, vol.~88, p.~234101, May 2002.

\bibitem{Bollt2002153}
E.~M. Bollt, L.~Billings, and I.~B. Schwartz, ``A manifold independent approach
  to understanding transport in stochastic dynamical systems,'' {\em Phys. D},
  vol.~173, no.~3–4, pp.~153 -- 177, 2002.

\bibitem{MatSchuss1982}
B.~J. Matkowsky and Z.~Schuss, ``Diffusion across characteristic boundaries,''
  {\em SIAM J. Appl. Math.}, vol.~42, no.~4, p.~822, 1982.

\bibitem{Weinan::towards2006}
W.~E and E.~Vanden-Eijnden, ``Towards a theory of transition paths,'' {\em J.
  Stat. Phys}, vol.~123, no.~3, pp.~503--523, 2006.

\bibitem{Eijnden::tpt2006}
E.~Vanden-Eijnden, ``Transition path theory,'' in {\em Computer Simulations in
  Condensed Matter Systems: From Materials to Chemical Biology} (M.~Ferrario,
  G.~Ciccotti, and K.~Binder, eds.), vol.~1, pp.~453--493, Springer, 2006.

\bibitem{Metzner::mjp2009}
P.~Metzner, C.~Sch{\"u}tte, and E.~Vanden-Eijnden, ``Transition path theory for
  {M}arkov jump processes,'' {\em Multiscale Model. Simul.}, vol.~7,
  p.~1192–1219, January 2009.

\bibitem{TPT2010}
W.~E and E.~Vanden-Eijnden, ``Transition-path theory and path-finding
  algorithms for the study of rare events,'' {\em Annu. Rev. Phys. Chem.},
  vol.~61, pp.~391--420, 2010.

\bibitem{Noe::folding2009}
F.~No{\'e}, C.~Sch{\"u}tte, E.~Vanden-Eijnden, L.~Reich, and T.~R. Weikl,
  ``Constructing the equilibrium ensemble of folding pathways from short
  off-equilibrium simulations,'' {\em Proc. Natl. Acad. Sci. U.S.A.}, vol.~106,
  no.~45, p.~19011–19016, 2009.

\bibitem{Cameron::LJcluster2014}
M.~Cameron and E.~Vanden-Eijnden, ``Flows in complex networks: Theory,
  algorithms, and application to {L}ennard-{J}ones cluster rearrangement,''
  {\em J. Stat. Phys.}, vol.~156, no.~3, pp.~427--454, 2014.

\bibitem{Ott::chaos1993}
E.~Ott, {\em Chaos in dynamical systems}.
\newblock Cambridge University Press, 1993.

\bibitem{subcrit2010}
W.~E, X.~Zhou, and X.~Cheng, ``Subcritical bifurcation in spatially extended
  systems,'' {\em Nonlinearity}, vol.~25, p.~761, 2012.

\bibitem{Lu::PTRF2015}
J.~Lu and J.~Nolen, ``Reactive trajectories and the transition path process,''
  {\em Probability Theory and Related Fields}, vol.~161, no.~1--2,
  pp.~195--244, 2015.

\bibitem{Ahuja1993NFT}
R.~K. Ahuja, T.~L. Magnanti, and J.~B. Orlin, {\em Network Flows: Theory,
  Algorithms, and Applications}.
\newblock Upper Saddle River, NJ, USA: Prentice-Hall, Inc., 1993.

\end{thebibliography}

\end{document}